%% file: main.tex
\documentclass[leqno]{amsart}

\usepackage[foot]{amsaddr}

\usepackage{color}
\usepackage{geometry}
\usepackage{subfiles}
\usepackage{graphicx}
\usepackage[section]{placeins}
\usepackage[hidelinks,pdfusetitle]{hyperref}
\usepackage{stmaryrd}
\SetSymbolFont{stmry}{bold}{U}{stmry}{m}{n}
\usepackage{amsmath,amssymb,amsfonts,amsthm,textcomp}
\usepackage{mathtools}
\usepackage{tensor}
\usepackage{multicol}
\usepackage{subfig}
\usepackage{enumitem}
\graphicspath{
              {./figures_section-1/}
              {./figures_section-2/}
              {./figures_section-3/}
             }

\input{commands}

\begin{document}

\title{A Unified Gradient Theory for Frame-Indifferent Rates of Tensorial Internal Variables}
\author{Luis Espath}
\address{School of Mathematical Sciences\\
University of Nottingham\\
Nottingham, NG7 2RD, United Kingdom}
\email{espath@gmail.com}

\date{\today}

\begin{abstract}
\noindent
\subfile{./abstract.tex}

\end{abstract}

\maketitle

\tableofcontents                        


\subfile{./section.tex}


\appendix

\subfile{./appendix.tex}


\footnotesize

\bibliographystyle{unsrt}
\bibliography{bib}

\end{document}

%% file: commands.tex
\newfont{\tenbfsl}{cmbxti9 scaled 1200}
\newfont{\tenbbb}{msbm10}
\newfont{\svnbbb}{msbm8}

\newcommand{\bs}[1]{\boldsymbol{#1}}
\newcommand{\cl}[1]{\mathcal{#1}}

\newcommand{\bb}[1]{\mathbb{#1}}
\newcommand{\br}[1]{\boldsymbol{\mathrm{{#1}}}}

\newcommand{\twovdots}{\mskip+2mu\colon\mskip-2mu}
\makeatletter
\def\threevdots{\mskip+4mu\vbox{\baselineskip2.25\p@ \lineskiplimit\z@
  \kern4.9\p@\hbox{.}\hbox{.}\hbox{.}}\mskip+3.8mu}
\makeatother

\newcommand{\twovc}{%
  \mskip8mu
  \mathord{%
    \mathclap{%
      \vcenter{%
        \offinterlineskip
        \halign{\hfil$##$\hfil\cr
          \noalign{\kern-2.5pt}
          \cdot\cr
          \noalign{\kern-2.5pt}
          \cdot\cr
          \noalign{\kern-0.5pt}
          \scriptscriptstyle\otimes\cr
        }%
      }%
    }%
  }%
  \mskip8mu
}
\newcommand{\twovcM}{%
  \mskip8mu
  \mathord{%
    \mathclap{%
      \vcenter{%
        \offinterlineskip
        \halign{\hfil$##$\hfil\cr
          \noalign{\kern-1.5pt}
          \cdot\cr
          \noalign{\kern-0.5pt}
          \scriptscriptstyle\otimes\cr
          \noalign{\kern-0.0pt}
          \cdot\cr
        }%
      }%
    }%
  }%
  \mskip8mu
}
\newcommand{\twovcR}{%
  \mskip8mu
  \mathord{%
    \mathclap{%
      \vcenter{%
        \offinterlineskip
        \halign{\hfil$##$\hfil\cr
          \scriptscriptstyle\otimes\cr
          \noalign{\kern-0.5pt}
          \cdot\cr
          \noalign{\kern-2.5pt}
          \cdot\cr
          \noalign{\kern-1.0pt}
        }%
      }%
    }%
  }%
  \mskip8mu
}

\newcommand{\surp}[1]{\{\mskip-6mu\{ \mskip-3mu{#1}\mskip-2mu \}\mskip-6mu\}}

\newcommand{\fr}[2]{\textstyle{\frac{{#1}}{{#2}}}}

\newcommand{\dv}{\,\mathrm{d}v}
\newcommand{\da}{\,\mathrm{d}a}
\newcommand{\ds}{\,\mathrm{d}s}

\newcommand{\prt}{\cl{P}}
\newcommand{\srf}{\cl{S}}
\newcommand{\edg}{\cl{C}}
\newcommand{\dprt}{\partial\cl{P}}
\newcommand{\ddprt}{\partial^2\cl{P}}

\newcommand{\intPt}{\int\limits_{\prt_t}}

\newcommand{\intdPt}{\int\limits_{\dprt_t}}
\newcommand{\intddPt}{\int\limits_{\ddprt_t}}
\newcommand{\intS}{\int\limits_{\cl{S}}}

\newcommand{\intC}{\int\limits_{\cl{C}}}

\newcommand{\trans}{\scriptscriptstyle\mskip-1mu\top\mskip-2mu}
\newcommand{\sperp}{\scriptscriptstyle\mskip-2mu\perp\mskip-2mu}

\newcommand{\rperp}{\rotatebox{90}{\(\scriptscriptstyle\mskip-0mu\top\mskip-0mu\)}}
\newcommand{\tr}{\mathrm{tr}\mskip2mu}
\newcommand{\sym}{\mathrm{sym}\mskip2mu}
\newcommand{\skw}{\mathrm{skw}\mskip2mu}

\newcommand{\Grad}{\mathrm{grad}\mskip2mu}
\newcommand{\Div}{\mathrm{div}\mskip2mu}

\newcommand{\Grads}{\Grad_{\mskip-2mu\scriptscriptstyle\cl{S}}}
\newcommand{\Divs}{\Div_{\mskip-5mu\scriptscriptstyle\cl{S}}}

\newcommand{\dd}[2]{\frac{\mathrm{d}{#1}}{\mathrm{d}{#2}}}

\newtheorem{thm}{Theorem}

\newtheorem{lem}{Lemma}
\newtheorem{prop}{Proposition}
\newtheorem{post}{Postulate}
\newtheorem{rmk}{Remark}
\newtheorem{set}{Definition}
\newtheorem{ass}{Assumption}

\newcommand{\vel}{\bs{\upsilon}}

\newcommand{\texand}{\qquad\text{and}\qquad}

\newcommand{\where}{\qquad\text{where }}
\newcommand{\with}{\qquad\text{with }}
\newcommand{\forevery}{\qquad\text{for every }}
\newcommand{\bsts}{\bs{t}_{\mskip-2mu\scriptscriptstyle\cl{S}}}
\newcommand{\bshs}{\bs{h}_{\mskip-2mu\scriptscriptstyle\cl{S}}}
\newcommand{\bsxs}{\br{X}_{\mskip-2mu\scriptscriptstyle\cl{S}}}
\newcommand{\bstc}{\bs{t}_{\mskip-2mu\scriptscriptstyle\cl{C}}}
\newcommand{\bsms}{\bs{m}_{\mskip-2mu\scriptscriptstyle\cl{S}}}

\newcommand{\ellcg}[2]{\ell_{({#1},{#2})}}

%% file: abstract.tex
We develop a thermodynamically consistent framework for weakly nonlocal continua with tensor-valued internal variables. The formulation is based on the family of generators \(\bs{\Gamma}_{\!\alpha} = \br{W} + \alpha \br{D}\), with \(\alpha \in \{0,1\}\), which unifies corotational and upper-convected transport. This construction induces a canonical frame-indifferent evolution for both the internal variable and its spatial gradient, thereby providing a closure for gradient-dependent theories.

Starting from the balance laws for linear momentum and microforces, together with an internal power expenditure depending on the internal variable and its gradient, we derive a local free-energy imbalance for incompressible isothermal processes. Under isotropy and inherited symmetry assumptions, the imbalance admits a canonical decomposition into contributions associated with the stretching tensor \(\br{D}\), the second gradient of the velocity \(\Grad \br{L}\), the frame-indifferent rate \(\mathfrak{D}_{\!\alpha}\br{J}\) induced by the generator \(\bs{\Gamma}_{\!\alpha}\), and its gradient \(\mathfrak{D}^{\!\nabla}_{\!\alpha}(\Grad \br{J})\). This decomposition yields constitutive restrictions ensuring thermodynamic consistency and identifies the induced higher-order stress contributions associated with gradient-dependent internal variables.

Finally, we construct a coupled gradient theory combining viscoelasticity and constrained orientational order, in which distinct internal variables evolve under different transport mechanisms within this unified thermodynamic setting. This establishes a consistent extension of classical continuum theories with tensorial internal variables, including Oldroyd-B- and Landau--de Gennes-type theories.
\\
\textbf{AMS subject classifications:}
$\cdot$
76A10 
$\cdot$
76A15 
$\cdot$
80A17 
$\cdot$
74A30 
$\cdot$
76A05 

%% file: section.tex
\section{Introduction}

Internal variables are customarily used in continuum mechanics to represent the effective behavior of materials with underlying microstructure, including complex fluids, polymeric systems, and geophysical and biological materials. These variables enrich the kinematical description by capturing internal degrees of freedom in addition to the classical primary fields. Their evolution is governed by constitutive relations subject to fundamental physical principles such as thermodynamic consistency and frame indifference. Classical rational thermodynamic theories with internal variables, as developed by Coleman and Gurtin \cite{Col67}, provide a framework in which the constitutive structure is derived from the Clausius--Duhem inequality via the Coleman--Noll procedure. In its classical local form, this approach does not by itself account for gradient-dependent effects or geometrically constrained rates.

The continuum theory developed by Gurtin, Polignone, \& Vi\~{n}als \cite{Gur96} for phase transitions in binary fluids provides a classical and prototypical example of a weakly nonlocal theory with a scalar internal variable, namely an order parameter whose free energy density depends on both the variable and its gradient. Likewise, the continuum theory of liquid crystals, as developed by Ericksen \cite{Eri61,Eri62} and extended by Leslie \cite{Les66,Les68}, constitutes a fundamental example of a theory with a vector-valued internal variable subject to nonlinear constraints. Under a variational lens, Ericksen's work has been further developed by Virga \cite{Vir18} and Sonnet \& Virga \cite[Chapter 3]{Son12}.

In the context of viscoelasticity, internal variables are often tensor-valued and describe memory and relaxation effects. The internal variables considered here are intrinsically symmetric contravariant \((2,0)\)-tensors and therefore transform contravariantly under changes of observer. Their evolution is governed by frame-indifferent transport laws, as in Oldroyd's invariant formulation of rheological equations of state \cite{Old50}. Different choices of transport law lead to distinct constitutive responses, depending on the kinematic action assigned to the internal variable. In particular, Oldroyd \cite{Old50} considers tensorial internal variables associated with deformation and introduces transport laws induced by the motion. For spatial conformation-type tensors, advection induced by motion gives rise to the upper-convected rate. In contrast, Mattos \cite{Mat98} advocates a corotational evolution for a tensorial internal variable representing relaxation. Thus, the internal variables considered by Oldroyd and Mattos have the same intrinsic tensorial character, namely that of symmetric contravariant \((2,0)\)-tensors, but differ in the kinematic action governing their transport. In the former case, the tensor is advected by the full velocity gradient, whereas in the latter case the advection is restricted to its skew-symmetric part. The particular formulations just described are local and do not account for gradient-dependent effects.

Tensor-valued internal variables also arise in other areas of continuum mechanics, where their physical interpretation differs. For instance, symmetric traceless tensors are used to describe orientational order in liquid crystals within the Landau--de Gennes theory, as developed by De Gennes \& Prost \cite{Deg93}, while the same tensorial descriptor appears as fabric tensors in the modeling of anisotropic microstructures such as granular materials, as described by Goodman \& Cowin \cite{Goo72}. These examples highlight that, although such internal variables share a common tensorial structure, their underlying physical nature may suggest distinct geometric transport mechanisms. Gradient-dependent free-energy densities of the form \(\psi = \hat{\psi}(\br{Q},\Grad \br{Q})\) are classical in the theory of liquid crystals for order tensor theories; see, for instance, Sonnet \& Virga \cite[Chapter 4]{Son12}. In that context, the dependence on \(\Grad \br{Q}\) gives rise to additional stress contributions (often referred to as Ericksen stresses) as well as microstresses of the form \(\partial_{\Grad \br{Q}}\psi\), and leads to Allen--Cahn-type evolution equations driven by the variational derivative of the free energy.

Our aim is to provide a unified framework that systematically couples the evolution of a tensor-valued internal variable with that of its gradient while retaining frame indifference and thermodynamic consistency across different choices of transport structure. To this end, we introduce a family of rate laws based on a generator of the form
\begin{equation}\label{eq:generator}
\bs{\Gamma}_{\!\alpha} \coloneqq \br{W} + \alpha \br{D}, \where \alpha \in \{0,1\}.
\end{equation}
Here, \(\br{L}\coloneqq\Grad\vel\) is the velocity gradient, \(\br{D}\coloneqq\sym\br{L}\) is the stretching tensor, and \(\br{W}\coloneqq\skw\br{L}\) is the spin tensor. The parameter \(\alpha\) distinguishes between corotational and contravariant advection. We show that this generator induces a consistent evolution for the internal variable and its gradient, namely
\begin{equation}\label{eq:frame.indifferent.rate.intro}
\mathfrak{D}_{\!\alpha}\br{J} \coloneqq \dot{\br{J}} - \bs{\Gamma}_{\!\alpha}\br{J} - \br{J}\bs{\Gamma}_{\!\alpha}^{\trans},
\end{equation}
and
\begin{equation}
\mathfrak{D}^{\!\nabla}_{\!\alpha}(\Grad \br{J}) \coloneqq (\Grad \br{J}){\dot{\vphantom{\br{J}}}} + (\Grad \br{J})\br{L} - ((\Grad \bs{\Gamma}_{\!\alpha})^{\!\trans}\br{J})^{\!\trans} - \bs{\Gamma}_{\!\alpha}\Grad \br{J} - ((\Grad \br{J})^{\!\trans}\bs{\Gamma}_{\!\alpha}^{\trans})^{\!\trans} - \br{J}\Grad(\bs{\Gamma}_{\!\alpha}^{\trans}),
\end{equation}
thereby extending frame-indifferent rate formulations to weakly nonlocal settings.\footnote{The generator \eqref{eq:generator} yielding \eqref{eq:frame.indifferent.rate.intro} appears in the literature as the Gordon--Schowalter/Johnson--Segalman rate, after Gordon \& Schowalter \cite[Equation (16) therein]{Gor72} and Johnson \& Segalman \cite[Equation (2.18) therein]{Joh77} where \(\alpha\) is introduced as a phenomenological parameter. The corotational member, corresponding to \(\alpha=0\), also arises in liquid-crystal hydrodynamics as the co-rotational specialization of the Beris--Edwards \(Q\)-tensor model; see Beris \& Edwards \cite{Ber94}. The general Beris--Edwards transport includes additional flow-alignment terms. To be faithful to the geometry of deformation described in what follows, we restrict attention to \(\alpha \in \{0,1\}\). Other authors have also used \(\alpha \in [-1,1]\).} Moreover, using a microforce balance for the internal variable, we derive a tensorial Allen--Cahn-type equation linked to each frame-indifferent rate of the tensor-valued internal variable.

The main result of this work is that the local free-energy imbalance admits a canonical decomposition into contributions associated with deformation and gradient effects, whose structure is preserved across all admissible choices of \(\alpha\), while the specific contributions depend on \(\alpha\). The resulting dissipation structure separates into distinct mechanisms involving the stretching tensor \(\br{D}\), the second gradient of the velocity \(\Grad \br{L}\), the frame-indifferent rate \(\mathfrak{D}_{\!\alpha}\br{J}\), and its gradient \(\mathfrak{D}^{\!\nabla}_{\!\alpha}(\Grad \br{J})\). Moreover, the appearance of \(\mathfrak{D}^{\!\nabla}_{\!\alpha}(\Grad \br{J})\) in the internal power expenditure induces Korteweg--Ericksen-type contributions in both the stress and hyperstress. Consequently, a gradient theory of the motion is intrinsic to the resulting constitutive structure.

Finally, we demonstrate the applicability of the proposed framework by constructing a coupled gradient theory for viscoelastic fluids with a constrained orientational order. In particular, we consider a conformation tensor governed by the upper-convected rate, as in classical Oldroyd-type models, together with an orientation tensor governed by the corotational (Jaumann) rate, corresponding to the co-rotational specialization of hydrodynamic Landau--de Gennes/\(Q\)-tensor models. Both internal variables are endowed with gradient-dependent free energies, leading to a unified weakly nonlocal formulation in which elastic, orientational, and interfacial effects are coupled. This example shows that distinct geometric transport mechanisms may coexist within a single frame-indifferent and thermodynamically consistent framework, with each internal variable evolving according to its intrinsic physical character.

The structure of this work is as follows. In Section \S\ref{sc:kinematics.geometry}, we present the kinematics and geometry of the motion, including integro-differential identities for nonsmooth surfaces. Special attention is paid to frame-indifferent rates. In Section \S\ref{sc:balance.laws}, we introduce the balance laws for mass, linear momentum, and microforces for isothermal processes. In Section \S\ref{sc:thermodynamic.framework}, we introduce and develop a unified thermodynamic framework. We then present the main result of this work, namely the canonical decomposition of the free-energy imbalance into contributions associated with deformation and gradient effects. In Section \S\ref{sc:final.model}, we gather the final set of equations, constitutive relations, and establish suitable natural and essential boundary conditions. In Section \S\ref{sc:gradient.OB.LD}, we present a coupled gradient theory for viscoelastic fluids and constrained orientational order. For convenience, we record the core notation used throughout this work.
\begin{align*}
\varrho &:\ \text{mass density}, &
\vel &:\ \text{velocity}, \\[3pt]
\bs{b} &:\ \text{specific body force}, &
\bs{1} &:\ \text{second-order identity tensor}, \\[3pt]
\br{T} &:\ \text{Cauchy stress tensor}, &
\br{S} &:\ \text{deviatoric stress}, \\[3pt]
\bb{T} &:\ \text{hyperstress tensor}, &
\bb{S} &:\ \text{deviatoric hyperstress}, \\[3pt]
\bs{\Xi} &:\ \text{internal microforce}, &
\bs{\Upsilon} &:\ \text{external microforce}, \\[3pt]
\bb{X} &:\ \text{microstress}, &
\br{J} &:\ \text{internal variable}, \\[3pt]
\br{D} &:\ \text{stretching tensor}, &
\br{W} &:\ \text{spin tensor}, \\[3pt]
p,\bs{p} &:\ \text{pressure and hyperpressure representatives}, &
\pi &:\ \text{bulk pressure}, \\[3pt]
\mathfrak{D}_{\!\alpha} &:\ \text{frame-indifferent rate}, &
\mathfrak{D}^{\!\nabla}_{\!\alpha} &:\ \text{frame-indifferent gradient rate}, \\[3pt]
\Grad, \Div &:\ \text{gradient and divergence}, &
\otimes &:\ \text{dyadic product}, \\[3pt]
\Grads, \Divs &:\ \text{surface gradient and divergence}, &
\wedge &:\ \text{wedge product}, \\[3pt]
\bs{a} \cdot \bs{b}, \, \br{A} \twovdots \br{B}, \, \bb{A} \threevdots \bb{B} &:\ \text{inner product}, &
\bb{A} \twovc \bb{B} &:\ (\bb{A})_{k\ell i} (\bb{B})_{k\ell j}, \\[3pt]
\bb{A} \twovcR \bb{B} &:\ (\bb{A})_{ik\ell} (\bb{B})_{jk\ell}, &
\bb{A} \twovcM \bb{B} &:\ (\bb{A})_{ki\ell} (\bb{B})_{kj\ell}, \\[3pt]
(\bb{A}^{\!\trans})_{ijk} & \coloneqq (\bb{A})_{ikj}, &
(\bb{A}^{\!\sperp})_{ijk} & \coloneqq (\bb{A})_{jik}, \\[3pt]
(\bb{A}^{\!\rperp})_{ijk} & \coloneqq (\bb{A})_{jki}, &
\dot{\varphi} &:\ \text{material time derivative of }\varphi, \\[3pt]
\bs{\Gamma}_{\!\alpha} &:\ \text{kinematic generator} & & \\[3pt]
\phantom{\bs{\Gamma}_{\!\alpha}} &\mskip22mu \text{of frame-indifferent transport.}
\end{align*}

\section{Kinematics \& Geometry}\label{sc:kinematics.geometry}

We consider a three-dimensional Euclidean point space \(\cl{E}\) with associated vector space \(\cl{V}\). Within an arbitrary spatial part \(\prt_t \subset \cl{E}\) evolving with the material with respect to an inertial frame, let \(\bs{y}\) represent the motion which describes the trajectory of material points such that \(\prt_t \coloneqq \bs{y}(\prt_{\scriptscriptstyle{0}})\), with \(\prt_{\scriptscriptstyle{0}}\) being the reference region of \(\prt_t\). Thus, a material point \(\bs{x} \in \prt_{\scriptscriptstyle{0}}\) is mapped to \(\bs{y}(\bs{x},t) \in \prt_t\). For each \(\bs{y} \in \prt_t\), the tangent space \(T_{\bs{y}}\prt_t\) is modeled on the vector space \(\cl{V}\). Also, important to what follows, we allow the boundary of the spatial part \(\prt_t\), namely \(\partial\prt_t\) to be endowed with a discontinuous outward normal field, thereby allowing for nonsmooth surfaces.

Next, we present integro-differential identities. First, we define the curvature tensor as the negative surface gradient of the unit normal, as
\begin{equation}\label{eq:curvature.tensor}
\br{K} \coloneqq -\Grads\bs{n},
\end{equation}
where \(\Grads\) is the surface gradient operator. We also need to define the mean curvature by
\begin{equation}\label{eq:mean.curvature}
K \coloneqq \fr{1}{2} \tr\br{K} = - \fr{1}{2} \Divs\bs{n}.
\end{equation}
Since in this work, we consider that surfaces \(\srf\) may be nonsmooth, we define edges \(\edg\) at these discontinuities. In particular, we consider a nonsmooth closed oriented surface \(\srf\) with outward unit normal \(\bs{n}\) and limiting outward unit tangent-normals \(\bs{\nu}^+\) and \(\bs{\nu}^-\) at \(\edg\). Owing to the lack of smoothness at the edge \(\edg\), the surface divergence theorem on \(\srf\) exhibits a surplus. Let the tensor field \(\br{A}\) be defined on \(\srf\). Then, the following identity holds
\begin{equation}\label{eq:nonsmooth.divs.theo.closed.S}
\intS \Divs(\br{A}\br{P}) \da = \intC\surp{\br{A}\bs{\nu}}\ds,
\end{equation}
where \(\surp{\br{A}\bs{\nu}} \coloneqq \br{A}\bs{\nu}^+ + \br{A}\bs{\nu}^-\) and \(\br{P} = \bs{1} - \bs{n}\otimes \bs{n}\) is the surface projection tensor onto the surface. Note that \(\br{P}\) is linear and idempotent, that is, \(\br{P}^2 = \br{P}\). With the identity
\begin{equation}\label{eq:identity.div.surf.tensor}
\Divs \br{A} = \Divs(\br{A}\br{P}) + \Divs(\br{A}\bs{n} \otimes \bs{n}) = \Divs (\br{A}\br{P}) - 2K \br{A}\bs{n},
\end{equation}
the surface divergence theorem \eqref{eq:nonsmooth.divs.theo.closed.S} reads
\begin{equation}\label{eq:nonsmooth.divs.theo.closed.S.tensors}
\intS \Divs \br{A} \da = \intS 2K\br{A}\bs{n} \da + \intC \surp{\br{A}\bs{\nu}} \ds.
\end{equation}

The motion \(\bs{y}\) induces a natural mapping between the reference and current configurations at the level of tangent spaces. The deformation gradient pushes forward material vectors to spatial vectors \(\br{F}(\bs{x},t) \colon T_{\bs{x}}\prt_{\scriptscriptstyle 0} \to T_{\bs{y}}\prt_t\) and defines, for each material point, a linear mapping, such that
\begin{equation}
\br{F} \coloneqq \Grad_{\bs{x}} \bs{y}.
\end{equation}
Note that \(\Grad_{\bs{x}}\) represents the gradient with respect to the reference coordinate \(\bs{x}\) while \(\Grad \coloneqq \Grad_{\bs{y}}\) the gradient with respect to the spatial coordinate \(\bs{y}(\bs{x},t)\). \(\br{F}\) is the tangent (push-forward) map of the motion, also known as a two-point tensor. Its inverse \(\br{F}^{-1}\) pulls spatial vectors back to the reference configuration, while its transpose inverse \(\br{F}^{-\trans}\) acts on covectors as the dual (pull-back) map. Accordingly, the kinematics described here can be interpreted geometrically as the transport of vectors and covectors induced by the motion. We assume that \(J \coloneqq \det \br{F} > 0\). Owing to this assumption, the motion is locally invertible. Also, the velocity field \(\dot{\bs{y}}\) may be  described as a function \(\vel\) of the spatial point \(\bs{y}(\bs{x},t) \in \prt_t\) for \(\bs{x} \in \prt_{\scriptscriptstyle{0}}\) and \(t\) such that
\begin{equation}
\vel(\bs{y}(\bs{x},t),t) \coloneqq \dot{\bs{y}} (\bs{x},t),
\end{equation}
holding \(\bs{x}\) fixed. The field \(\vel\) represents the spatial description of the velocity. Therefore, using the chain rule, we have that for a sufficiently smooth spatial field \(\varphi\)
\begin{align}\label{eq:material.derivative}
\dot{\varphi}(\bs{y}(\bs{x},t),t) &= \partial_t \varphi(\bs{y}(\bs{x},t),t) + \Grad \varphi(\bs{y}(\bs{x},t),t) \cdot \dot{\bs{y}}(\bs{x},t) \nonumber\\[4pt]
&= \partial_t \varphi + (\Grad \varphi) \cdot \vel,
\end{align}
and deduce that the material derivative depends on the velocity \(\vel\). Next, in view of the material derivative \eqref{eq:material.derivative}, for the position vector \(\bs{r} \coloneqq \bs{y} - \bs{o}\) relative to an arbitrary chosen origin \(\bs{o}\), we have that
\begin{equation}\label{eq:material.derivative.position}
\dot{\bs{r}} = \partial_t \bs{r} + (\Grad \bs{r}) \vel = \vel.
\end{equation}
Additionally, the velocity gradient \(\br{L}(\bs{y},t) \colon T_{\bs{y}}\prt_t \to T_{\bs{y}}\prt_t\) is defined through
\begin{equation}\label{eq:velocity.gradient}
\br{L} \coloneqq \dot{\br{F}} \br{F}^{-1},
\end{equation}
with the consequence that
\begin{equation}\label{eq:velocity.gradient.alternate}
\br{L} = \Grad \vel.
\end{equation}

For completeness, we recall the transport laws used below. A detailed geometric derivation of the covariant, contravariant, and corotational rates from the modes of advection induced by the motion is provided in Appendix \ref{ap:transport.rates}. In particular, a spatial vector field \(\bs{a}\in T_{\bs{y}}\prt_t\), a \((1,0)\)-tensor, advects as a tangent if
\begin{equation}
\dot{\bs{a}}=\br{L}\bs{a}.
\end{equation}
The corresponding dual covector field \(\bs{\alpha}\in T^\ast_{\bs{y}}\prt_t\), a \((0,1)\)-tensor, evolves according to
\begin{equation}
\dot{\bs{\alpha}}=-\br{L}^{\!\trans}\bs{\alpha}.
\end{equation}
Consequently, for a contravariant second-order tensor
\begin{equation}
\br{G}\in T_{\bs{y}}\prt_t\otimes T_{\bs{y}}\prt_t,
\end{equation}
that is, a \((2,0)\)-tensor acting on two spatial covectors, contravariant advection is characterized by the upper-convected rate
\begin{equation}
\overset{\triangledown}{\br{G}} \coloneqq \dot{\br{G}}-\br{L}\br{G}-\br{G}\br{L}^{\!\trans}.
\end{equation}
Conversely, if the same contravariant tensor advects corotationally under the spin \(\br{W}=\skw\br{L}\), the corresponding corotational rate is
\begin{equation}
\overset{\circ}{\br{G}} \coloneqq \dot{\br{G}}-\br{W}\br{G}+\br{G}\br{W}.
\end{equation}
Thus, at the geometric level, both the upper-convected and corotational rates may be motivated by actions on the same contravariant \((2,0)\)-tensorial structure. The distinction lies in the chosen mode of advection. In the Euclidean tensor representation adopted for the continuum theory below, these actions motivate the generator
\begin{equation}
\bs{\Gamma}_{\!\alpha}=\br{W}+\alpha\br{D}, \qquad \alpha\in\{0,1\},
\end{equation}
which recovers the corotational rate for \(\alpha=0\) and the upper-convected rate for \(\alpha=1\).

\subsection{Frame-indifferent rates}\label{sc:frame.indifferent.rates}

Capitalizing on the work by Gordon \& Schowalter \cite{Gor72}, Johnson \& Segalman \cite{Joh77}, and Beris and Edwards \cite{Ber94}, we introduce a general class of internal-variable models based on frame-indifferent rates, encompassing both upper-convected and corotational formulations. Their geometric origins are developed in Appendix \ref{ap:transport.rates}.
\begin{set}[Frame-indifferent rate]
Let \(\br{J}(\bs{y},t) \in \mathrm{Sym}\)\footnote{\(\mathrm{Sym}\) is the space of symmetric second-order tensors on the ambient Euclidean vector space.} be an internal variable. For \(\alpha \in \{0,1\}\), define the frame-indifferent rate
\begin{equation}\label{eq:frame.indifferent.rate}
\mathfrak{D}_{\!\alpha}\br{J} = \dot{\br{J}} - \bs{\Gamma}_{\!\alpha}\br{J} - \br{J}\bs{\Gamma}_{\!\alpha}^{\trans}, \with \bs{\Gamma}_{\!\alpha} \coloneqq \br{W} + \alpha \br{D},
\end{equation}
induced by \(\bs{\Gamma}_{\!\alpha}\), where \(\br{D} \coloneqq \sym \br{L}\) and \(\br{W} \coloneqq \skw \br{L}\) and
\begin{equation}
\bs{\Gamma}_{\!0} \in \mathfrak{so}(3) \subset \mathfrak{gl}(3) \texand \bs{\Gamma}_{\!1} \in \mathfrak{gl}(3).
\end{equation}
Here, \(\mathfrak{gl}(3)=\bb{R}^{3\times 3}\) and \(\mathfrak{so}(3)=\{\bs{\Omega}\in\bb{R}^{3\times 3}\colon \bs{\Omega}^{\!\trans}=-\bs{\Omega}\}\) are the Lie algebras of \(\mathrm{GL}^+(3)\) and \(\mathrm{SO}(3)\), respectively; their tangent-space characterization is recalled in Appendix \ref{ap:alg.grp}. In the general case where \(\br{D}\neq\bs{0}\), one has \(\bs{\Gamma}_{\!1}\notin\mathfrak{so}(3)\). The rate is induced through the action of the generator on symmetric second-order tensors.
\end{set}
The frame-indifferent evolution of \(\Grad \br{J}\) compatible with \(\mathfrak{D}_{\!\alpha}\br{J}\) is induced by the generator \(\bs{\Gamma}_{\!\alpha}\) through the requirement of commutation with spatial differentiation, namely
\begin{equation}\label{eq:gradient.frame.indifferent.rate}
\mathfrak{D}^{\!\nabla}_{\!\alpha}(\Grad \br{J}) \coloneqq \Grad(\mathfrak{D}_{\!\alpha}\br{J}),
\end{equation}
such that
\begin{equation}
\mathfrak{D}^{\!\nabla}_{\!\alpha}(\Grad \br{J}) = \Grad \dot{\br{J}} - \Grad(\bs{\Gamma}_{\!\alpha}\br{J}) - \Grad(\br{J}\bs{\Gamma}_{\!\alpha}^{\trans}).
\end{equation}
This construction yields \(\mathfrak{D}^{\!\nabla}_{\!\alpha}(\Grad \br{J}) = \Grad(\mathfrak{D}_{\!\alpha}\br{J})\), enforcing consistency between the evolution of \(\br{J}\) and its spatial gradient.
\begin{rmk}[Generator-induced gradient rate]
The expression \eqref{eq:gradient.frame.indifferent.rate} is induced by the generator \(\bs{\Gamma}_{\!\alpha}\) underlying the frame-indifferent rate \(\mathfrak{D}_{\!\alpha}\br{J}\). It does not introduce an independent frame-indifferent rate for third-order tensors; instead, it enforces kinematic compatibility by lifting the action of \(\bs{\Gamma}_{\!\alpha}\) to spatial gradients. This construction propagates frame indifference from \(\br{J}\) to \(\Grad \br{J}\) and, at the thermodynamic level, gives rise to Korteweg--Ericksen-type stress contributions.
\end{rmk}

Equivalently, the frame-indifferent rate \(\mathfrak{D}_{\!\alpha}\br{J}\) can be expressed as
\begin{equation}\label{eq:frame.indifferent.rate.definition}
\mathfrak{D}_{\!\alpha} \br{J} \coloneqq \dot{\br{J}} - \br{W}\br{J} + \br{J}\br{W} - \alpha \big(\br{D}\br{J} + \br{J}\br{D}).
\end{equation}
Analogously, using the kinematic identity
\begin{equation}
\Grad \dot{\br{J}} = (\Grad \br{J}){\dot{\vphantom{\br{J}}}} + (\Grad \br{J})\br{L},
\end{equation}
together with the product rules
\begin{equation}
\Grad(\bs{\Gamma}_{\!\alpha}\br{J}) = ((\Grad \bs{\Gamma}_{\!\alpha})^{\!\trans}\br{J})^{\!\trans} + \bs{\Gamma}_{\!\alpha}\Grad \br{J} \texand \Grad(\br{J}\bs{\Gamma}_{\!\alpha}^{\trans}) = ((\Grad \br{J})^{\!\trans}\bs{\Gamma}_{\!\alpha}^{\trans})^{\!\trans} + \br{J}\Grad(\bs{\Gamma}_{\!\alpha}^{\trans}),
\end{equation}
which reads explicitly as
\begin{equation}
\big[\Grad(\bs{\Gamma}_{\!\alpha}\br{J})\big]_{ijk} = \partial_k (\bs{\Gamma}_{\!\alpha})_{i\ell} (\br{J})_{\ell j} + (\bs{\Gamma}_{\!\alpha})_{i\ell} \partial_k (\br{J})_{\ell j} \texand \big[\Grad(\br{J}\bs{\Gamma}_{\!\alpha}^{\trans})\big]_{ijk} = \partial_k (\br{J})_{i\ell} (\bs{\Gamma}_{\!\alpha})_{j\ell} + (\br{J})_{i\ell}  \partial_k (\bs{\Gamma}_{\!\alpha})_{j\ell},
\end{equation}
the frame-indifferent rate of \(\Grad \br{J}\) \eqref{eq:gradient.frame.indifferent.rate} may be written in the explicit form
\begin{equation}\label{eq:gradient.frame.indifferent.rate.explicit}
\mathfrak{D}^{\!\nabla}_{\!\alpha}(\Grad \br{J}) = (\Grad \br{J}){\dot{\vphantom{\br{J}}}} + (\Grad \br{J})\br{L} - ((\Grad \bs{\Gamma}_{\!\alpha})^{\!\trans}\br{J})^{\!\trans} - \bs{\Gamma}_{\!\alpha}\Grad \br{J} - ((\Grad \br{J})^{\!\trans}\bs{\Gamma}_{\!\alpha}^{\trans})^{\!\trans} - \br{J}\Grad(\bs{\Gamma}_{\!\alpha}^{\trans}).
\end{equation}
Here, \(\Grad\br{J}\) denotes the spatial gradient of \(\br{J}\), with the differentiation direction represented by its last index. Thus, \(\Grad\br{J}\) is a third-order Euclidean tensor with symmetry in its first two indices and
\begin{equation}
(\Grad \br{J})_{ijk}=\partial_k(\br{J})_{ij}.
\end{equation}

\begin{rmk}[Euclidean tensor-space convention]\label{rmk:euclidean.tensor.convention}
The preceding geometric discussion is used solely to motivate the kinematic actions underlying the frame-indifferent rates. From this point onward, we adopt the standard Euclidean tensor-space convention of continuum mechanics. Thus, second-order tensors are regarded as linear transformations of the ambient Euclidean vector space, while higher-order tensors belong to the corresponding Euclidean tensor spaces. Whenever appropriate, tensor products are compositions, and transpose, trace, determinant, and tensor inner products are the Euclidean operations. Fr\'echet derivatives and power-conjugate quantities are identified with their Euclidean Riesz representatives and are denoted by the same symbols. Accordingly, all power pairings below are represented by the standard complete tensor contractions \(\twovdots\) and \(\threevdots\).
\end{rmk}

\section{Balance Laws}\label{sc:balance.laws}

Here, we recall the balance laws of continuum mechanics in their spatial form.

\begin{post}[Mass balance]\label{pt:mass.balance}
The balance of mass
\begin{equation}
\dfrac{\mathrm{d}}{\mathrm{d}t} \intPt \varrho \dv_t = 0,
\end{equation}
holds for every spatial part \(\prt_t\), and localizes to
\begin{equation}
\dot{\varrho} + \varrho \, \Div \vel = 0.
\end{equation}
In the incompressible case considered herein, \(\varrho = \text{const}\), and therefore
\begin{equation}
\Div \vel = 0.
\end{equation}
\end{post}

\begin{post}[Linear \& angular momentum balance]\label{pt:linear.angular.momentum}
The balance of linear momentum
\begin{equation}\label{eq:partwise.linear.momentum}
\dfrac{\mathrm{d}}{\mathrm{d}t} \intPt \varrho \vel \dv_t = \intPt \varrho \bs{b} \dv_t + \intdPt \bsts \da_t + \intddPt \bstc \ds_t,
\end{equation}
holds for every spatial part \(\prt_t\). Here, \(\bs{b}\) denotes the body force per unit mass, the surface traction \(\bsts\) is the force per unit area exerted on the boundary \(\partial \prt_t\) given by
\begin{equation}\label{eq:surface.traction}
\bsts = (\br{T} - \Div \bb{T})\bs{n} - \Divs ((\bb{T}\bs{n}) \br{P}),
\end{equation}
where \(\br{T}\) is the Cauchy stress tensor and \(\bb{T}\) is the hyperstress tensor, and the edge traction \(\bstc\) is the force per unit length exerted on the edge \(\partial^2 \prt_t\) given by
\begin{equation}\label{eq:edge.traction}
\bstc = \surp{(\bb{T}\bs{n})\bs{\nu}}.
\end{equation}
The complete derivation of the surface and edge tractions along with the partwise balances can be found in the work by Espath \cite{Esp23} and Fried and Gurtin \cite{Fri06}. The interested reader is also referred to the work by Fosdick \cite{Fos16}, which provides a comprehensive treatment of the subject. Then, using the divergence and surface divergence theorem on nonsmooth surfaces on the partwise linear momentum balance \eqref{eq:partwise.linear.momentum} and standard localization arguments, this reduces to the pointwise condition
\begin{equation}\label{eq:pointwise.linear.momentum}
\varrho \dot{\vel} = \Div (\br{T} - \Div \bb{T}) + \varrho \bs{b}.
\end{equation}

In addition to \eqref{eq:pointwise.linear.momentum}, the balance of angular momentum
\begin{equation}\label{eq:partwise.angular.momentum}
\dfrac{\mathrm{d}}{\mathrm{d}t} \intPt \bs{r} \wedge \varrho \vel \dv_t = \intPt \bs{r} \wedge \varrho \bs{b} \dv_t + \intdPt (\bs{r} \wedge \bsts + \bsms) \da_t + \intddPt \bs{r} \wedge \bstc \ds_t,
\end{equation}
holds for every spatial part \(\prt_t\). Here, \(\bs{p} \wedge \bs{q} = \bs{p} \otimes \bs{q} - \bs{q} \otimes \bs{p}\) denotes the wedge product and \(\bsms\) is the surface couple traction, which is the moment per unit area exerted on the boundary \(\partial \prt_t\), and takes the form
\begin{equation}\label{eq:surface.couple.traction}
\bsms = \bs{n} \wedge (\bb{T}\bs{n}) \bs{n}.
\end{equation}
Using the divergence and surface divergence theorem on nonsmooth surfaces on the partwise angular momentum balance \eqref{eq:partwise.angular.momentum} and standard localization arguments, this reduces to the pointwise condition
\begin{equation}
\skw \br{T} = \bs{0},
\end{equation}
that is, the Cauchy stress tensor is symmetric
\begin{equation}
\br{T} = \br{T}^{\trans}.
\end{equation}
\end{post}

\begin{post}[Microstructural balance]\label{pt:microstructural.momentum}
The balance of microstructural interactions
\begin{equation}
\intPt (\bs{\Xi} + \bs{\Upsilon}) \dv_t + \intdPt \bsxs \da_t = \bs{0},
\end{equation}
holds for every spatial part \(\prt_t\). Here, \(\bsxs\) is the microstructural traction, which is the force per unit area exerted on the boundary \(\partial \prt_t\) given by
\begin{equation}\label{eq:microstructural.traction}
\bsxs \coloneqq \bb{X}\bs{n},
\end{equation}
while \(\bs{\Xi}\) and \(\bs{\Upsilon}\) are the internal and external microstructural body forces per unit volume, respectively, and \(\bb{X}\) is the microstress tensor. Localization yields
\begin{equation}\label{eq:pointwise.microstructural.linear.momentum}
\bs{\Xi} + \bs{\Upsilon} + \Div \bb{X} = \bs{0}.
\end{equation}
\end{post}

\begin{thm}[Power identity]
If Postulates \ref{pt:mass.balance}, \ref{pt:linear.angular.momentum}, and \ref{pt:microstructural.momentum} hold for every spatial part \(\prt_t\), then the power identity
\begin{equation}
\cl{W}_{\mathrm{int}}(\cl{P}_t) = \cl{W}_{\mathrm{ext}}(\cl{P}_t) - \dot{\cl{K}}(\cl{P}_t),
\end{equation}
holds for every spatial part \(\prt_t\), where we identify the total internal power expended on the spatial part \(\prt_t\) as
\begin{equation}\label{eq:W.int}
\cl{W}_{\mathrm{int}}(\cl{P}_t) = \intPt \br{T} \twovdots \br{D} \dv_t + \intPt \bb{T} \threevdots \Grad \br{L} \dv_t - \intPt \bs{\Xi} \twovdots \mathfrak{D}_{\!\alpha}\br{J} \dv_t + \intPt \bb{X} \threevdots \mathfrak{D}^{\!\nabla}_{\!\alpha}(\Grad \br{J}) \dv_t,
\end{equation}
the total external power expended on the spatial part \(\prt_t\) as
\begin{align}\label{eq:W.ext}
\cl{W}_{\mathrm{ext}}(\cl{P}_t) ={}& \intPt \bs{\Upsilon} \twovdots \mathfrak{D}_{\!\alpha}\br{J} \dv_t + \intPt \varrho \bs{b} \cdot \vel \dv_t + \intdPt (\bsts \cdot \vel + \bshs \cdot \partial_n \vel) \da_t + \intdPt \bsxs \twovdots \mathfrak{D}_{\!\alpha}\br{J} \da_t + \intddPt \bstc \cdot \vel \ds_t,
\end{align}
where \(\partial_n \vel \coloneqq (\Grad \vel) \bs{n}\) denotes the normal derivative of the velocity field and \(\bshs\) is the hypertraction defined as
\begin{equation}
\bshs \coloneqq (\bb{T}\bs{n}) \bs{n},
\end{equation}
and the variation of the kinetic energy of the spatial part \(\prt_t\) as
\begin{equation}\label{eq:K.var}
\dot{\cl{K}}(\cl{P}_t) = {\fr{1}{2}} \dfrac{\mathrm{d}}{\mathrm{d}t} \intPt \varrho |\vel|^2 \dv_t.
\end{equation}
The remaining tractions \(\bsts\), \(\bstc\), and \(\bsxs\) are respectively given by \eqref{eq:surface.traction}, \eqref{eq:edge.traction}, and \eqref{eq:microstructural.traction}.
\end{thm}
\begin{proof}
Testing the linear momentum balance \eqref{eq:pointwise.linear.momentum} against the velocity \(\vel\), integrating over \(\cl{P}_t\), and using the divergence and surface divergence theorems, we obtain
\begin{align}
\intPt \varrho \dot{\vel} \cdot \vel \dv_t ={}& \intPt (\Div (\br{T} - \Div \bb{T})) \cdot \vel \dv_t + \intPt \varrho \bs{b} \cdot \vel \dv_t \nonumber\\[4pt]
={}& - \intPt \br{T} \twovdots \Grad \vel \dv_t - \intPt \bb{T} \threevdots \Grad^2 \vel \dv_t + \intPt \varrho \bs{b} \cdot \vel \dv_t \nonumber\\[4pt]
&+ \intdPt (((\br{T} - \Div \bb{T})\bs{n} - \Divs ((\bb{T}\bs{n}) \br{P})) \cdot \vel + ((\bb{T}\bs{n})\bs{n} \otimes \bs{n}) \twovdots \Grad \vel) \da_t \nonumber\\[4pt]
&+ \intddPt \surp{(\bb{T}\bs{n})\bs{\nu}} \cdot \vel \ds_t.
\end{align}
Rearranging, we obtain the power balance
\begin{align}
\intPt \br{T} \twovdots \Grad \vel \dv_t + \intPt \bb{T} \threevdots \Grad^2 \vel \dv_t ={}& \intPt \varrho \bs{b} \cdot \vel \dv_t - {\fr{1}{2}} \dfrac{\mathrm{d}}{\mathrm{d}t} \intPt \varrho |\vel|^2 \dv_t \nonumber\\[4pt]
&+ \intdPt (((\br{T} - \Div \bb{T})\bs{n} - \Divs ((\bb{T}\bs{n}) \br{P})) \cdot \vel + ((\bb{T}\bs{n})\bs{n} \otimes \bs{n}) \twovdots \Grad \vel) \da_t \nonumber\\[4pt]
&+ \intddPt \surp{(\bb{T}\bs{n})\bs{\nu}} \cdot \vel \ds_t.
\end{align}

In this generalized theory, we postulate that the microstructural interactions expend power through the frame-indifferent rate \(\mathfrak{D}_{\!\alpha}\br{J}\). Accordingly, testing the microstructural linear momentum balance \eqref{eq:pointwise.microstructural.linear.momentum} against \(\mathfrak{D}_{\!\alpha}\br{J}\) yields the microstructural contribution to the power expenditure
\begin{align}
- \intPt \bs{\Xi} \twovdots \mathfrak{D}_{\!\alpha}\br{J} \dv_t + \intPt \bb{X} \threevdots \Grad(\mathfrak{D}_{\!\alpha}\br{J}) \dv_t = \intPt \bs{\Upsilon} \twovdots \mathfrak{D}_{\!\alpha}\br{J} \dv_t + \intdPt \bb{X}\bs{n} \twovdots \mathfrak{D}_{\!\alpha}\br{J} \da_t.
\end{align}
Therefore, we identify the total external power as
\begin{align}
\cl{W}_{\mathrm{ext}}(\cl{P}_t) ={}& \intPt \bs{\Upsilon} \twovdots \mathfrak{D}_{\!\alpha}\br{J} \dv_t + \intPt \varrho \bs{b} \cdot \vel \dv_t \nonumber\\[4pt]
&+ \intdPt (((\br{T} - \Div \bb{T})\bs{n} - \Divs ((\bb{T}\bs{n}) \br{P})) \cdot \vel + ((\bb{T}\bs{n})\bs{n} \otimes \bs{n}) \twovdots \Grad \vel) \da_t \nonumber\\[4pt]
&+ \intdPt \bb{X}\bs{n} \twovdots \mathfrak{D}_{\!\alpha}\br{J} \da_t  + \intddPt \surp{(\bb{T}\bs{n})\bs{\nu}} \cdot \vel \ds_t,
\end{align}
with the respective traction definitions. The internal power and kinetic energy variation are then identified as the remaining contribution along with \eqref{eq:gradient.frame.indifferent.rate}, yielding \eqref{eq:W.int} and \eqref{eq:K.var}.
\end{proof}

\section{A unified thermodynamic framework for frame-indifferent rates of weakly nonlocal theories}\label{sc:thermodynamic.framework}

It is natural to identify the conjugate forces appearing in the internal power \eqref{eq:W.int} with their power-active representatives, since any component orthogonal to the corresponding kinematic quantity is powerless.
\begin{ass}[Power-active symmetries]\label{as:inherited.symmetries}
Owing to the power-conjugate pairs appearing in the internal power \eqref{eq:W.int}, we identify the constitutively determined parts of \(\br{T}\) and \(\bb{T}\), as well as \(\bs{\Xi}\) and \(\bb{X}\), with the projections having the same power-active symmetries as \(\br{D}\), \(\Grad \br{L}\), \(\mathfrak{D}_{\!\alpha}\br{J}\), and \(\mathfrak{D}^{\!\nabla}_{\!\alpha}(\Grad \br{J})\), respectively. The reactive Lagrange multipliers associated with incompressibility are treated separately below.
\end{ass}
In particular, for an incompressible motion, \(\br{D}\) belongs to the space
\begin{equation}\label{eq:Sym0.space}
\mathrm{Sym}_0 \coloneqq \{\br{A} \colon \br{A}=\br{A}^{\!\trans},\ \tr\br{A}=0\},
\end{equation}
whereas \(\Grad\br{L}\) belongs to the space
\begin{equation}\label{eq:E.space}
\cl{A} \coloneqq \{\bb{A} \colon (\bb{A})_{ijk}=(\bb{A})_{ikj},\ (\bb{A})_{iik}=0\}.
\end{equation}
For every second-order tensor \(\br{C}\), define
\begin{equation}\label{eq:P0.def}
\bb{P}_0[\br{C}] \coloneqq \sym(\br{C})-\fr{1}{3}\tr(\br{C})\bs{1}.
\end{equation}
Then, \(\bb{P}_0\) is the orthogonal projection onto \(\mathrm{Sym}_0\). In particular,
\begin{equation}\label{eq:P0.power}
\bb{P}_0[\br{C}]\twovdots\br{D}=\sym(\br{C})\twovdots\br{D}.
\end{equation}
Here and below, \(\sym_{23}\bb{C} \coloneqq \fr{1}{2}(\bb{C}+\bb{C}^{\trans})\) denotes symmetrization with respect to the last two indices.
For every third-order tensor \(\bb{C}\) symmetric with respect to its last two indices, define
\begin{equation}\label{eq:P.A.def}
(\bb{P}_{\!\!\scriptscriptstyle\cl{A}}[\bb{C}])_{ijk} \coloneqq (\bb{C})_{ijk} - \fr{1}{4}\big(\delta_{ij}(\bb{C})_{mmk}+\delta_{ik}(\bb{C})_{mmj}\big),
\end{equation}
where \(\delta_{ij}\) is the Kronecker delta. Then, \(\bb{P}_{\!\!\scriptscriptstyle\cl{A}}\) is the orthogonal projection onto \(\cl{A}\). In particular,
\begin{equation}\label{eq:P.A.power}
\bb{P}_{\!\!\scriptscriptstyle\cl{A}}[\bb{C}] \threevdots \Grad\br{L} = \bb{C} \threevdots \Grad\br{L}.
\end{equation}
Additionally, we assume that the response functions are isotropic in their dependence on the kinematic quantities.
\begin{ass}[Response functions]\label{as:response.functions}
Let the free-energy density be given by the response function
\begin{equation}
\psi = \hat{\psi}(\br{J}, \Grad \br{J}),
\end{equation}
where \(\partial_{\br{J}}\psi\) and \(\partial_{\Grad\br{J}}\psi\) denote the Euclidean Riesz representatives of the Fr\'echet derivatives with respect to \(\br{J}\) and \(\Grad\br{J}\), respectively. Since \(\Grad\br{J}\) is symmetric in its first two indices, \(\partial_{\Grad\br{J}}\psi\) denotes the Euclidean Riesz representative with the same symmetry. Assume that \(\hat{\psi}\) is isotropic in its dependence on \(\br{J} \in \mathrm{Sym}\) and \(\Grad \br{J}\). Moreover, assume the stresses decomposition
\begin{equation}\label{eq:stress.decomposition}
\br{T} = -p \bs{1} + \br{S} \texand \bb{T} = - \sym_{23}(\bs{1} \otimes \bs{p}) + \bb{S},
\end{equation}
where \(p\) and \(\bs{p}\) are the pressure and hyperpressure, respectively, and \((\bs{1}\otimes\bs{p})_{ijk}=\delta_{ij}p_k\). Thus,
\begin{equation}
\sym_{23}(\bs{1}\otimes\bs{p})\threevdots\Grad\br{L}=\bs{p}\cdot\Grad\Div\vel=0.
\end{equation}
The tensor \(\sym_{23}(\bs{1}\otimes\bs{p})\) is therefore the representative of the Lagrange multiplier, symmetric in its last two indices, associated with \(\Grad\Div\vel=\bs{0}\). In addition,
\begin{equation}
\Div\big(\Div(\sym_{23}(\bs{1}\otimes\bs{p}))\big)=\Grad\Div\bs{p},
\end{equation}
so the bulk momentum balance depends on the pressure and hyperpressure only through \(p-\Div\bs{p}\). Consequently, \(p\) and \(\bs{p}\) are Lagrange-multiplier representatives determined up to this bulk gauge, with their representatives further constrained by the boundary conditions. Moreover, \(\br{S}\in\mathrm{Sym}_0\) and \(\bb{S}\in\cl{A}\) are the total deviatoric stress and total deviatoric hyperstress, respectively. Lastly, the internal microforce and microstress are given by the response functions
\begin{equation}
\bs{\Xi} \coloneqq \hat{\bs{\Xi}}\big(\br{J},\Grad\br{J},\mathfrak{D}_{\!\alpha}\br{J},\mathfrak{D}^{\!\nabla}_{\!\alpha}(\Grad\br{J})\big) \texand \bb{X} \coloneqq \hat{\bb{X}}\big(\br{J},\Grad\br{J},\mathfrak{D}_{\!\alpha}\br{J},\mathfrak{D}^{\!\nabla}_{\!\alpha}(\Grad\br{J})\big).
\end{equation}
\end{ass}

\begin{lem}[Infinitesimal consequences of isotropy]\label{lem:infinitesimal.isotropy}
Consider Assumption \ref{as:response.functions}. Then, the following hold. First,
\begin{equation}\label{eq:dpsi.sym}
\partial_{\br{J}}\psi \in \mathrm{Sym}.
\end{equation}
Second, for every skew-symmetric tensor \(\bs{\Omega}\in\mathfrak{so}(3)\),
\begin{equation}\label{eq:infinitesimal.isotropy.full}
\partial_{\br{J}}\psi \twovdots (\bs{\Omega}\br{J} - \br{J}\bs{\Omega}) + \partial_{\Grad\br{J}}\psi \threevdots (\bs{\Omega}\Grad\br{J} + ((\Grad\br{J})^{\!\trans}\bs{\Omega}^{\!\trans})^{\!\trans} - (\Grad\br{J})\bs{\Omega}) = 0.
\end{equation}
Equivalently,
\begin{equation}\label{eq:infinitesimal.isotropy.equivalent}
\skw(\partial_{\br{J}}\psi\br{J} - \br{J}\partial_{\br{J}}\psi + \partial_{\Grad\br{J}}\psi \twovcR \Grad\br{J} + \partial_{\Grad\br{J}}\psi \twovcM \Grad\br{J} - \Grad\br{J} \twovc \partial_{\Grad\br{J}}\psi) = \bs{0}.
\end{equation}
Here, the tensor contraction notations \(\twovcR\), \(\twovcM\), and \(\twovc\) denote the tensor products of the first, second, and third indices, respectively, of a third-order tensor with a third-order tensor, that is, \((\bb{A} \twovcR \bb{B})_{ij} \coloneqq (\bb{A})_{ik\ell}(\bb{B})_{jk\ell}\), \((\bb{A} \twovcM \bb{B})_{ij} \coloneqq (\bb{A})_{ki\ell}(\bb{B})_{kj\ell}\), \((\bb{A} \twovc \bb{B})_{ij} \coloneqq (\bb{A})_{k\ell i}(\bb{B})_{k\ell j}\).
\end{lem}
\begin{proof}
Since the first argument of \(\hat{\psi}\) belongs to \(\mathrm{Sym}\), its Fr\'echet derivative with respect to \(\br{J}\) is a linear functional on \(\mathrm{Sym}\). Its Euclidean Riesz representative therefore belongs to \(\mathrm{Sym}\), which proves \eqref{eq:dpsi.sym}.

Next, let \(\br{Q}(\epsilon)\in\mathrm{SO}(3)\) be a smooth one-parameter family such that
\begin{equation}\label{eq:Q.family}
\br{Q}(0)=\bs{1}, \qquad \dfrac{\mathrm d}{\mathrm d\epsilon}\bigg|_{\epsilon=0}\br{Q}(\epsilon) = \bs{\Omega}, \texand \bs{\Omega}^{\!\trans} = -\bs{\Omega}.
\end{equation}
By isotropy of \(\hat{\psi}\),
\begin{equation}\label{eq:isotropy.statement}
\hat{\psi}\big(\br{Q}\br{J}\br{Q}^{\!\trans}, \br{Q}*(\Grad\br{J})\big) = \hat{\psi}(\br{J},\Grad\br{J}) \qquad \text{for every } \br{Q}\in\mathrm{SO}(3),
\end{equation}
where the natural action of \(\br{Q}\) on the third-order Euclidean tensor \(\Grad\br{J}\) is induced by the standard rotational action on each index. Thus,
\begin{equation}\label{eq:Q.action.gradJ}
(\br{Q}*(\Grad\br{J}))^{ijk} = (\br{Q})^{i}{}_{p}(\br{Q})^{j}{}_{q}(\br{Q})^{k}{}_{r}(\Grad\br{J})^{pqr}.
\end{equation}
Equivalently, in components,
\begin{equation}
(\br{Q}*(\Grad\br{J}))_{ijk} = (\br{Q})_{ip}(\br{Q})_{jq}(\br{Q})_{kr}(\Grad\br{J})_{pqr}.
\end{equation}
Differentiating \eqref{eq:isotropy.statement} with respect to \(\epsilon\) at \(\epsilon=0\), we obtain
\begin{equation}\label{eq:isotropy.diff}
0 = \partial_{\br{J}}\psi \twovdots \dfrac{\mathrm d}{\mathrm d\epsilon}\bigg|_{\epsilon=0} (\br{Q}\br{J}\br{Q}^{\!\trans}) + \partial_{\Grad\br{J}}\psi \threevdots \dfrac{\mathrm d}{\mathrm d\epsilon}\bigg|_{\epsilon=0} (\br{Q}*(\Grad\br{J})).
\end{equation}
Now, differentiating the first term in \eqref{eq:isotropy.diff}, we have that
\begin{equation}\label{eq:QJQ.diff}
\dfrac{\mathrm d}{\mathrm d\epsilon}\bigg|_{\epsilon=0} (\br{Q}\br{J}\br{Q}^{\!\trans}) = \bs{\Omega}\br{J} + \br{J}\bs{\Omega}^{\!\trans} = \bs{\Omega}\br{J} - \br{J}\bs{\Omega}.
\end{equation}
Also, differentiating the second term in \eqref{eq:isotropy.diff} yields the infinitesimal rotational action on the three indices of \(\Grad\br{J}\):
\begin{equation}\label{eq:QgradJ.diff.index}
\dfrac{\mathrm d}{\mathrm d\epsilon}\bigg|_{\epsilon=0} (\br{Q}*(\Grad\br{J}))_{ijk} = (\bs{\Omega})_{i\ell}\partial_k(\br{J})_{\ell j} + (\bs{\Omega})_{j\ell} \partial_k (\br{J})_{i\ell} + (\bs{\Omega})_{k\ell} \partial_\ell(\br{J})_{ij}.
\end{equation}
Accordingly, the infinitesimal action may be written compactly as
\begin{equation}\label{eq:QgradJ.diff.compact}
\dfrac{\mathrm d}{\mathrm d\epsilon}\bigg|_{\epsilon=0} (\br{Q}*(\Grad\br{J})) = \bs{\Omega} \, \Grad\br{J} + ((\Grad\br{J})^{\!\trans}\bs{\Omega}^{\trans})^{\!\trans} - (\Grad\br{J})\bs{\Omega}.
\end{equation}
Substituting \eqref{eq:QJQ.diff} and \eqref{eq:QgradJ.diff.compact} into \eqref{eq:isotropy.diff} proves \eqref{eq:infinitesimal.isotropy.full}. Collecting the coefficient of the arbitrary skew-symmetric tensor \(\bs{\Omega}\) proves \eqref{eq:infinitesimal.isotropy.equivalent}.
\end{proof}
An analogous result to Lemma \ref{lem:infinitesimal.isotropy} is derived by Sonnet \& Virga \cite[Equation (4.29)]{Son12} by requiring the elastic energy density to be frame indifferent.

The following lemma provides the adjoint identities for the Lyapunov and commutator operators used below.
\begin{lem}[Adjoint identities]\label{lem:lyapunov.commutator.identities}
Let \(\br{J}\in\mathrm{Sym}\), and define the Lyapunov operator \(\cl{L}_{\!\br{J}}\colon\mathrm{Sym}\to\mathrm{Sym}\) by
\begin{equation}\label{eq:lyapunov.operator}
\cl{L}_{\!\br{J}}[\br{A}] \coloneqq \br{A}\br{J}+\br{J}\br{A}.
\end{equation}
Then, \(\cl{L}_{\!\br{J}}\) is self-adjoint with respect to the Frobenius inner product, that is,
\begin{equation}\label{eq:lyapunov.self.adjoint}
\br{B}\twovdots\cl{L}_{\!\br{J}}[\br{A}]=\cl{L}_{\!\br{J}}[\br{B}]\twovdots\br{A}\qquad\text{for every }\br{A},\br{B}\in\mathrm{Sym}.
\end{equation}
Define also the commutator operator \(\cl{C}_{\!\br{J}}\colon\mathfrak{so}(3)\to\mathrm{Sym}\) by
\begin{equation}
\cl{C}_{\!\br{J}}[\bs{\Omega}]\coloneqq\bs{\Omega}\br{J}-\br{J}\bs{\Omega}.
\end{equation}
Its adjoint \(\cl{C}_{\!\br{J}}^\ast\colon\mathrm{Sym}\to\mathfrak{so}(3)\) is given by
\begin{equation}
\cl{C}_{\!\br{J}}^\ast[\br{B}]\coloneqq\br{B}\br{J}-\br{J}\br{B},
\end{equation}
and hence, for every \(\br{B}\in\mathrm{Sym}\) and \(\bs{\Omega}\in\mathfrak{so}(3)\),
\begin{equation}\label{eq:commutator.adjoint}
\br{B}\twovdots\cl{C}_{\!\br{J}}[\bs{\Omega}]
=\cl{C}_{\!\br{J}}^\ast[\br{B}]\twovdots\bs{\Omega}.
\end{equation}
Consequently, if \(\bs{\Gamma}_{\!\alpha}=\br{W}+\alpha\br{D}\), where \(\br{D}\in\mathrm{Sym}\) and \(\br{W}\in\mathfrak{so}(3)\), then
\begin{align}\label{eq:generator.adjoint}
\br{B}\twovdots(\bs{\Gamma}_{\!\alpha}\br{J} + \br{J}\bs{\Gamma}_{\!\alpha}^{\trans}) &= \alpha \, \cl{L}_{\!\br{J}}[\br{B}]\twovdots\br{D} + \cl{C}_{\!\br{J}}^\ast[\br{B}]\twovdots\br{W} \nonumber\\[4pt]
&= \alpha \, (\br{B}\br{J}+\br{J}\br{B})\twovdots\br{D} + (\br{B}\br{J}-\br{J}\br{B})\twovdots\br{W}.
\end{align}
\end{lem}
\begin{proof}
For every \(\br{A},\br{B},\br{J}\in\mathrm{Sym}\), cyclicity of the trace gives
\begin{align}
\br{B}\twovdots\cl{L}_{\!\br{J}}[\br{A}]
&=\tr(\br{B}\br{A}\br{J})+\tr(\br{B}\br{J}\br{A}) \nonumber\\[4pt]
&=\tr(\br{J}\br{B}\br{A})+\tr(\br{B}\br{J}\br{A}) \nonumber\\[4pt]
&=(\br{B}\br{J}+\br{J}\br{B})\twovdots\br{A}
=\cl{L}_{\!\br{J}}[\br{B}]\twovdots\br{A},
\end{align}
which proves \eqref{eq:lyapunov.self.adjoint}. Similarly, using \(\br{B}^{\!\trans}=\br{B}\), \(\br{J}^{\!\trans}=\br{J}\), and \(\bs{\Omega}^{\!\trans}=-\bs{\Omega}\),
\begin{align}
\br{B}\twovdots(\bs{\Omega}\br{J}-\br{J}\bs{\Omega})
&=\tr(\br{B}\bs{\Omega}\br{J})-\tr(\br{B}\br{J}\bs{\Omega}) \nonumber\\[4pt]
&=\tr((\br{J}\br{B}-\br{B}\br{J})\bs{\Omega})
=(\br{B}\br{J}-\br{J}\br{B})\twovdots\bs{\Omega},
\end{align}
which proves \eqref{eq:commutator.adjoint}. Finally,
\begin{equation}
\bs{\Gamma}_{\!\alpha}\br{J}+\br{J}\bs{\Gamma}_{\!\alpha}^{\trans}
=\alpha(\br{D}\br{J}+\br{J}\br{D})+\br{W}\br{J}-\br{J}\br{W},
\end{equation}
so \eqref{eq:generator.adjoint} follows from \eqref{eq:lyapunov.self.adjoint} and \eqref{eq:commutator.adjoint}.
\end{proof}

\begin{lem}[Gradient-generator adjoint identity]\label{lem:gradient.generator.adjoint}
Let \(\br{J}\in\mathrm{Sym}\), let \(\bb{Q}\) be a third-order tensor symmetric with respect to its first two indices, and define
\begin{equation}\label{eq:H.generic}
\bb{H}[\bb{Q},\br{J}] \coloneqq ((\bb{Q}^{\trans}\br{J})^{\!\trans})+(\br{J}\bb{Q})^{\sperp}.
\end{equation}
Then, for every sufficiently smooth second-order tensor field \(\bs{\Gamma}\),
\begin{equation}\label{eq:gradient.generator.adjoint.general}
\bb{Q}\threevdots\big(((\Grad\bs{\Gamma})^{\!\trans}\br{J})^{\!\trans}+\br{J}\Grad(\bs{\Gamma}^{\trans})\big)
=\bb{H}[\bb{Q},\br{J}]\threevdots\Grad\bs{\Gamma}.
\end{equation}
Consequently, for \(\alpha \in \{0,1\}\), if
\begin{equation}
\bs{\Gamma}_{\!\alpha}=\fr{1+\alpha}{2}\br{L}+\fr{\alpha-1}{2}\br{L}^{\!\trans}
\end{equation}
and
\begin{equation}\label{eq:H.alpha.generic}
\bb{H}_{\alpha}[\bb{Q},\br{J}] \coloneqq \fr{1+\alpha}{2}\bb{H}[\bb{Q},\br{J}]+\fr{\alpha-1}{2}\big(\bb{H}[\bb{Q},\br{J}]\big)^{\sperp},
\end{equation}
then
\begin{equation}\label{eq:gradient.generator.adjoint.alpha}
\bb{Q}\threevdots\big(((\Grad\bs{\Gamma}_{\!\alpha})^{\!\trans}\br{J})^{\!\trans}+\br{J}\Grad(\bs{\Gamma}_{\!\alpha}^{\trans})\big)
=\bb{H}_{\alpha}[\bb{Q},\br{J}]\threevdots\Grad\br{L}.
\end{equation}
\end{lem}
\begin{proof}
By the definitions of the third-order transpose and tensor products,
\begin{equation}
\big[((\Grad\bs{\Gamma})^{\!\trans}\br{J})^{\!\trans}\big]_{ijk}=\partial_k(\bs{\Gamma})_{i\ell}(\br{J})_{\ell j}
\texand
\big[\br{J}\Grad(\bs{\Gamma}^{\trans})\big]_{ijk}=(\br{J})_{i\ell}\partial_k(\bs{\Gamma})_{j\ell}.
\end{equation}
Therefore, using \((\br{J})_{ij}=(\br{J})_{ji}\) and relabeling dummy indices,
\begin{align}
(\bb{Q})_{ijk}\big(\partial_k(\bs{\Gamma})_{i\ell}(\br{J})_{\ell j}+(\br{J})_{i\ell}\partial_k(\bs{\Gamma})_{j\ell}\big) ={}& \big((\bb{Q})_{i\ell k}(\br{J})_{j\ell}+(\br{J})_{j\ell}(\bb{Q})_{\ell i k}\big)\partial_k(\bs{\Gamma})_{ij} \nonumber\\[4pt]
={}&\big(\bb{H}[\bb{Q},\br{J}]\big)_{ijk}\partial_k(\bs{\Gamma})_{ij},
\end{align}
which proves \eqref{eq:gradient.generator.adjoint.general}. Moreover,
\begin{equation}
\Grad\bs{\Gamma}_{\!\alpha}=\fr{1+\alpha}{2}\Grad\br{L}+\fr{\alpha-1}{2}(\Grad\br{L})^{\sperp}.
\end{equation}
Since \(\bb{A}\threevdots\bb{B}^{\sperp}=\bb{A}^{\sperp}\threevdots\bb{B}\) for every pair of third-order tensors \(\bb{A}\) and \(\bb{B}\), \eqref{eq:gradient.generator.adjoint.general} yields
\begin{align}
\bb{H}[\bb{Q},\br{J}]\threevdots\Grad\bs{\Gamma}_{\!\alpha}
={}&\big(\fr{1+\alpha}{2}\bb{H}[\bb{Q},\br{J}]+\fr{\alpha-1}{2}\big(\bb{H}[\bb{Q},\br{J}]\big)^{\sperp}\big)\threevdots\Grad\br{L} \nonumber\\[4pt]
={}&\bb{H}_{\alpha}[\bb{Q},\br{J}]\threevdots\Grad\br{L},
\end{align}
which proves \eqref{eq:gradient.generator.adjoint.alpha}.
\end{proof}

\begin{thm}[Canonical decomposition of the pointwise free-energy imbalance]\label{thm:canonical.free.energy.imbalance}
Consider Assumptions \ref{as:inherited.symmetries} and \ref{as:response.functions}. Additionally, assume that, for every spatial part \(\prt_t\) evolving with the material with respect to an inertial frame in an incompressible and isothermal fluid, the partwise free-energy imbalance
\begin{equation}
\dot{\overline{\intPt \psi \dv_t}} \le \cl{W}_{\mathrm{int}},
\end{equation}
holds with the internal power expenditure \(\cl{W}_{\mathrm{int}}\) given by \eqref{eq:W.int}. Define the second-order energetic stress by
\begin{align}\label{eq:H2.energetic}
\br{H}_{\alpha} \coloneqq{}& \alpha (\br{J}\partial_{\br{J}}\psi + \partial_{\br{J}}\psi\br{J}) \nonumber\\[4pt]
&+ \sym\big(\alpha (\partial_{\Grad\br{J}}\psi \twovcR \Grad \br{J} + \partial_{\Grad\br{J}}\psi \twovcM \Grad \br{J}) - \Grad\br{J} \twovc \partial_{\Grad\br{J}}\psi\big),
\end{align}
and define the third-order energetic hyperstress by
\begin{equation}\label{eq:H.energetic}
\bb{H}_{\alpha} \coloneqq \fr{1+\alpha}{2}\bb{H}+\fr{\alpha-1}{2}\bb{H}^{\sperp} \with \bb{H} \coloneqq ((\partial_{\Grad\br{J}}\psi)^{\trans}\br{J})^{\!\trans}+(\br{J}\partial_{\Grad\br{J}}\psi)^{\sperp}.
\end{equation}
Then, the following pointwise free-energy imbalance holds for every \(\alpha \in \{0,1\}\)
\begin{align}\label{eq:pointwise.dissipation.final.reduced.equivalent}
0 \ge{}& - (\br{S} - \bb{P}_0[\br{H}_{\alpha}]) \twovdots \br{D} - (\bb{S} - \bb{P}_{\!\!\scriptscriptstyle\cl{A}}[\sym_{23}\bb{H}_{\alpha}]) \threevdots \Grad \br{L} \nonumber\\[4pt]
&+ (\partial_{\br{J}}\psi + \bs{\Xi}) \twovdots \mathfrak{D}_{\!\alpha}\br{J} + (\partial_{\Grad\br{J}}\psi-\bb{X}) \threevdots \mathfrak{D}^{\!\nabla}_{\!\alpha}(\Grad \br{J}).
\end{align}
\end{thm}
\begin{proof}
By the Reynolds transport theorem and incompressibility,
\begin{equation}
\dfrac{\mathrm{d}}{\mathrm{d}t}\intPt\psi\dv_t=\intPt\dot{\psi}\dv_t.
\end{equation}
The chain rule and the definitions of the frame-indifferent rates give
\begin{align}\label{eq:free.energy.rate.indifferent}
\dot{\psi} ={}& \partial_{\br{J}}\psi \twovdots \mathfrak{D}_{\!\alpha} \br{J} + \partial_{\Grad \br{J}}\psi \threevdots \mathfrak{D}^{\!\nabla}_{\!\alpha}(\Grad \br{J}) \nonumber\\[4pt]
&+ \partial_{\br{J}}\psi \twovdots (\bs{\Gamma}_{\!\alpha}\br{J} + \br{J}\bs{\Gamma}_{\!\alpha}^{\trans}) - \partial_{\Grad \br{J}}\psi \threevdots ((\Grad\br{J})\br{L}) \nonumber\\[4pt]
&+ \partial_{\Grad \br{J}}\psi \threevdots (\Grad(\bs{\Gamma}_{\!\alpha}\br{J}) + \Grad(\br{J}\bs{\Gamma}_{\!\alpha}^{\trans})),
\end{align}
and since \(\br{J},\partial_{\br{J}}\psi\in\mathrm{Sym}\), Lemma \ref{lem:lyapunov.commutator.identities}, with \(\br{B}=\partial_{\br{J}}\psi\), gives
\begin{equation}\label{eq:second.order.transport.split}
\partial_{\br{J}}\psi \twovdots (\bs{\Gamma}_{\!\alpha}\br{J}+\br{J}\bs{\Gamma}_{\!\alpha}^{\trans}) = \alpha(\br{J}\partial_{\br{J}}\psi+\partial_{\br{J}}\psi\br{J})\twovdots\br{D} +(\partial_{\br{J}}\psi\br{J}-\br{J}\partial_{\br{J}}\psi)\twovdots\br{W}.
\end{equation}
The terms that do not involve \(\Grad\bs{\Gamma}_{\!\alpha}\) satisfy
\begin{multline}\label{eq:gradient.transport.split}
\partial_{\Grad\br{J}}\psi\threevdots\big(\bs{\Gamma}_{\!\alpha}\Grad\br{J}+((\Grad\br{J})^{\!\trans}\bs{\Gamma}_{\!\alpha}^{\trans})^{\!\trans}-(\Grad\br{J})\br{L}\big) = \\[4pt]
\sym\big(\alpha(\partial_{\Grad\br{J}}\psi\twovcR\Grad\br{J} + \partial_{\Grad\br{J}}\psi\twovcM\Grad\br{J})-\Grad\br{J}\twovc\partial_{\Grad\br{J}}\psi\big)\twovdots\br{D} \\[4pt]
+ \skw\big(\partial_{\Grad\br{J}}\psi\twovcR\Grad\br{J} + \partial_{\Grad\br{J}}\psi\twovcM\Grad\br{J}-\Grad\br{J}\twovc\partial_{\Grad\br{J}}\psi\big)\twovdots\br{W}.
\end{multline}
Applying Lemma \ref{lem:gradient.generator.adjoint} with \(\bb{Q}=\partial_{\Grad\br{J}}\psi\), the terms involving the gradient of the generator obey
\begin{equation}\label{eq:gradient.generator.split}
\partial_{\Grad\br{J}}\psi\threevdots\big(((\Grad\bs{\Gamma}_{\!\alpha})^{\!\trans}\br{J})^{\!\trans} + \br{J}\Grad(\bs{\Gamma}_{\!\alpha}^{\trans})\big) = \bb{H}_{\alpha}\threevdots\Grad\br{L}.
\end{equation}
Substituting \eqref{eq:second.order.transport.split}, \eqref{eq:gradient.transport.split}, and \eqref{eq:gradient.generator.split} into the partwise free-energy imbalance, using \eqref{eq:W.int}, and localizing yield
\begin{align}
0 \ge{}& - (\br{S} - \br{H}_{\alpha}) \twovdots \br{D} \nonumber\\[4pt]
&+ (\partial_{\br{J}}\psi\br{J} - \br{J}\partial_{\br{J}}\psi + \skw (\partial_{\Grad\br{J}}\psi \twovcR \Grad \br{J} + \partial_{\Grad\br{J}}\psi \twovcM \Grad \br{J} - \Grad\br{J} \twovc \partial_{\Grad\br{J}}\psi)) \twovdots \br{W} \nonumber\\[4pt]
&- (\bb{S} - \bb{H}_{\alpha}) \threevdots \Grad \br{L} \nonumber\\[4pt]
&+ (\partial_{\br{J}}\psi + \bs{\Xi}) \twovdots \mathfrak{D}_{\!\alpha}\br{J} + (\partial_{\Grad\br{J}}\psi-\bb{X}) \threevdots \mathfrak{D}^{\!\nabla}_{\!\alpha}(\Grad \br{J}).
\end{align}
The second term vanishes identically by \eqref{eq:infinitesimal.isotropy.equivalent}. Moreover, in view of \eqref{eq:P0.power} and \eqref{eq:P.A.power}, since \(\br{D}\in\mathrm{Sym}_0\) and \(\Grad\br{L}\in\cl{A}\), only the projections \(\bb{P}_0[\br{H}_{\alpha}]\) and \(\bb{P}_{\!\!\scriptscriptstyle\cl{A}}[\sym_{23}\bb{H}_{\alpha}]\) expend power. This proves \eqref{eq:pointwise.dissipation.final.reduced.equivalent}.
\end{proof}

Restricting attention to uncoupled dissipative mechanisms, a sufficient set of constitutive restrictions ensuring that the inequality \eqref{eq:pointwise.dissipation.final.reduced.equivalent} holds for every admissible process is given by
\begin{equation}\label{eq:constraint}
\left\{
\begin{aligned}
&(\br{S} - \bb{P}_0[\br{H}_{\alpha}]) \twovdots \br{D} \ge 0, \\[4pt] & (\bb{S} - \bb{P}_{\!\!\scriptscriptstyle\cl{A}}[\sym_{23}\bb{H}_{\alpha}]) \threevdots \Grad \br{L} \ge 0, \\[4pt] & - (\partial_{\br{J}}\psi + \bs{\Xi}) \twovdots \mathfrak{D}_{\!\alpha}\br{J} \ge 0, \\[4pt] & - (\partial_{\Grad\br{J}}\psi-\bb{X}) \threevdots \mathfrak{D}^{\!\nabla}_{\!\alpha}(\Grad \br{J}) \ge 0.
\end{aligned}
\right.
\end{equation}
Therefore, from the thermodynamic constraint \eqref{eq:constraint}\(_1\), we stipulate that
\begin{equation}
\br{S} = 2 \mu \br{D} + \bb{P}_0[\br{H}_{\alpha}],
\end{equation}
where \(\mu \ge 0\) is the shear viscosity. Moreover, with this choice, the dissipation \eqref{eq:constraint}\(_1\) becomes
\begin{equation}
2 \mu |\br{D}|^2 \ge 0.
\end{equation}
Next, from the thermodynamic constraint \eqref{eq:constraint}\(_2\), we stipulate that
\begin{equation}
\bb{S} = \bb{S}^\mathrm{dis} + \bb{P}_{\!\!\scriptscriptstyle\cl{A}}[\sym_{23}\bb{H}_{\alpha}],
\end{equation}
where \(\bb{S}^\mathrm{dis}\) is a dissipative contribution satisfying
\begin{equation}
\bb{S}^\mathrm{dis} \threevdots \Grad \br{L} \ge 0.
\end{equation}
This requires a positive semidefinite linear operator acting on \(\Grad \br{L}\). To this end, let \(\bb{S}^{\mathrm{dis}}\) be defined through a linear isotropic operator acting on \(\Grad \br{L}\), that is,
\begin{equation}
\bb{S}^{\mathrm{dis}} = \pmb{\bb{C}} (\Grad \br{L}),
\end{equation}
where \(\pmb{\bb{C}}\) is a sixth-order isotropic tensor with the symmetries induced by the inner product, including the symmetry on the last two indices of \(\Grad \br{L}\). Its general representation involves five independent coefficients (see Espath \cite[Chapter 5]{Esp23}), that is, for every \(\bb{A}\) with the same last-two-index symmetry as \(\Grad \br{L}\), namely \((\bb{A})_{ijk}=(\bb{A})_{ikj}\),
\begin{align}\label{eq:6.tensor.isotropic.action}
(\pmb{\bb{C}})_{ijk\ell mn} (\bb{A})_{\ell mn} = {}&\ellcg{6}{1}(\delta_{ij} (\bb{A})_{kmm} + \delta_{ik} (\bb{A})_{jmm} + 2\delta_{jk} (\bb{A})_{mmi})\nonumber\\[4pt]
&+ 2\ellcg{6}{2}(\delta_{ij} (\bb{A})_{mmk}+\delta_{ik} (\bb{A})_{mmj})\nonumber\\[4pt]
&+ 2\ellcg{6}{3}((\bb{A})_{kij} + (\bb{A})_{jik})\nonumber\\[4pt]
&+ 2\ellcg{6}{4}(\bb{A})_{ijk} + \ellcg{6}{5}\delta_{jk} (\bb{A})_{imm}.
\end{align}
For the continuum type of theory involving incompressible materials, consider that \(\bb{A}\) satisfies \(\bs{1} \twovdots \bb{A} = (\bb{A})_{jji} =\bs{0}\), then, expressions \eqref{eq:6.tensor.isotropic.action} specialize to
\begin{align}\label{eq:6.tensor.isotropic.action.traceless.A}
(\pmb{\bb{C}})_{ijk\ell mn} (\bb{A})_{\ell mn} = {}&\ellcg{6}{1}(\delta_{ij} (\bb{A})_{kmm} + \delta_{ik} (\bb{A})_{jmm})\nonumber\\[4pt]
&+2\ellcg{6}{3}((\bb{A})_{kij}+(\bb{A})_{jik})\nonumber\\[4pt]
&+2\ellcg{6}{4}(\bb{A})_{ijk}+\ellcg{6}{5}\delta_{jk} (\bb{A})_{imm}.
\end{align}
If, in addition, one requires that \(\bb{S}^{\mathrm{dis}} = \pmb{\bb{C}} \bb{A}\) is to satisfy \(\bs{1} \twovdots \bb{S}^{\mathrm{dis}} = \bs{0}\), then the coefficients must obey
\begin{equation}
4\ellcg{6}{1} + 2\ellcg{6}{3} + \ellcg{6}{5} = 0,
\end{equation}
and the action \eqref{eq:6.tensor.isotropic.action.traceless.A} is specialized to
\begin{align}\label{eq:6.tensor.isotropic.action.traceless}
(\pmb{\bb{C}})_{ijk\ell mn} (\bb{A})_{\ell mn} ={}& 2\ellcg{6}{3}((\bb{A})_{kij}+(\bb{A})_{jik} - \fr{1}{4}(\delta_{ij}(\bb{A})_{kmm}+\delta_{ik}(\bb{A})_{jmm})) \nonumber\\[4pt]
&+ 2\ellcg{6}{4}(\bb{A})_{ijk} + \ellcg{6}{5}(\delta_{jk}(\bb{A})_{imm} - \fr{1}{4}(\delta_{ij}(\bb{A})_{kmm}+\delta_{ik}(\bb{A})_{jmm})).
\end{align}
Additionally, the inner product of \(\bb{A}\) with \(\pmb{\bb{C}}\bb{A}\) is given by
\begin{align}\label{eq:6.tensor.isotropic.inner.product.traceless}
\cl{Q}_{\pmb{\bb{C}}}(\bb{A})\coloneqq{}& (\bb{A})_{ijk}(\pmb{\bb{C}})_{ijk\ell mn} (\bb{A})_{\ell mn} \nonumber\\[4pt]
={}& 2\ellcg{6}{3}((\bb{A})_{ijk}(\bb{A})_{kij}+(\bb{A})_{ijk}(\bb{A})_{jik}) \nonumber\\[4pt]
&+ 2\ellcg{6}{4}(\bb{A})_{ijk}(\bb{A})_{ijk}+\ellcg{6}{5}(\bb{A})_{ikk}(\bb{A})_{imm}.
\end{align}
Also, from the thermodynamic constraint \eqref{eq:constraint}\(_3\), we stipulate that
\begin{equation}
\bs{\Xi} = \bs{\Xi}^\mathrm{dis} - \partial_{\br{J}}\psi,
\end{equation}
where \(\bs{\Xi}^\mathrm{dis}\) is a dissipative contribution satisfying
\begin{equation}
- \bs{\Xi}^\mathrm{dis} \twovdots \mathfrak{D}_{\!\alpha}\br{J} \ge 0.
\end{equation}
Thus, we set
\begin{equation}\label{eq:Xi.dis}
\bs{\Xi}^\mathrm{dis} = - \beta \, \mathfrak{D}_{\!\alpha}\br{J},
\end{equation}
where \(\beta \ge 0\) is an internal-variable drag coefficient, that is, the inverse of a mobility. Lastly, from the thermodynamic constraint \eqref{eq:constraint}\(_4\), we stipulate that
\begin{equation}
\bb{X} = \partial_{\Grad\br{J}}\psi.
\end{equation}

\begin{lem}[Harmonic decomposition of admissible third-order tensors]\label{lem:harmonic.decomp.A}
Let \(\cl{A}\) be the space defined in \eqref{eq:E.space}. Then, \(\cl{A}\) has dimension \(15\) and admits the orthogonal decomposition
\begin{equation}\label{eq:A.decomposition}
\cl{A} = \cl{H}^7 \oplus \cl{H}^5 \oplus \cl{H}^3,
\end{equation}
where \(\cl{H}^7\), \(\cl{H}^5\), and \(\cl{H}^3\) are \(\mathrm{SO}(3)\)-irreducible subspaces of dimensions \(7\), \(5\), and \(3\), respectively. Equivalently, every \(\bb{A} \in \cl{A}\) admits the unique orthogonal decomposition
\begin{equation}\label{eq:A.harmonic.decomposition}
\bb{A}=\bb{A}^{(7)}+\bb{A}^{(5)}+\bb{A}^{(3)},
\end{equation}
where the superscripts refer to the dimensions of the corresponding harmonic subspaces.

Define
\begin{equation}\label{eq:def.ai.lemma}
a_i \coloneqq (\bb{A})_{imm}.
\end{equation}
Then, the three components are given by
\begin{equation}\label{eq:A7.mode.lemma}
(\bb{A}^{(7)})_{ijk} \coloneqq \fr{1}{3} ((\bb{A})_{ijk}+(\bb{A})_{jik}+(\bb{A})_{kij}) - \fr{1}{15} (\delta_{ij}a_k+\delta_{ik}a_j+\delta_{jk}a_i),
\end{equation}
\begin{equation}\label{eq:A3.mode.lemma}
(\bb{A}^{(3)})_{ijk} \coloneqq \fr{2}{5} \delta_{jk}a_i - \fr{1}{10}(\delta_{ij}a_k+\delta_{ik}a_j),
\end{equation}
and
\begin{equation}\label{eq:A5.mode.lemma}
(\bb{A}^{(5)})_{ijk} \coloneqq (\bb{A})_{ijk}-(\bb{A}^{(7)})_{ijk} - (\bb{A}^{(3)})_{ijk}.
\end{equation}
Equivalently,
\begin{equation}\label{eq:A5.mode.explicit.lemma}
(\bb{A}^{(5)})_{ijk} = \fr{2}{3} (\bb{A})_{ijk} - \fr{1}{3} ((\bb{A})_{jik}+(\bb{A})_{kij}) +\fr{1}{6} (\delta_{ij}a_k+\delta_{ik}a_j) - \fr{1}{3} \, \delta_{jk}a_i.
\end{equation}

These tensors satisfy
\begin{align}
&(\bb{A}^{(7)})_{ijk}=(\bb{A}^{(7)})_{(ijk)},\qquad
&(\bb{A}^{(7)})_{iik}=(\bb{A}^{(7)})_{imm}=0,\qquad
&\label{eq:A7.props.lemma}\\[4pt]
&(\bb{A}^{(5)})_{ijk}=(\bb{A}^{(5)})_{ikj},\qquad
&(\bb{A}^{(5)})_{iik}=(\bb{A}^{(5)})_{imm}=0,\qquad
&(\bb{A}^{(5)})_{ijk}+(\bb{A}^{(5)})_{jik}+(\bb{A}^{(5)})_{kij}=0,\label{eq:A5.props.lemma}\\[4pt]
&(\bb{A}^{(3)})_{ijk}=(\bb{A}^{(3)})_{ikj},\qquad
&(\bb{A}^{(3)})_{iik}=0,\qquad
&(\bb{A}^{(3)})_{imm}=a_i.\label{eq:A3.props.lemma}
\end{align}
Moreover, the decomposition \eqref{eq:A.harmonic.decomposition} is orthogonal with respect to the Euclidean inner product, that is,
\begin{equation}\label{eq:orthogonality.modes.lemma}
\bb{A}^{(7)} \threevdots \bb{A}^{(5)} = \bb{A}^{(7)} \threevdots \bb{A}^{(3)} = \bb{A}^{(5)} \threevdots \bb{A}^{(3)} = 0.
\end{equation}
\end{lem}
\begin{proof}
Since \((\bb{A})_{ijk}=(\bb{A})_{ikj}\), the tensor \(\bb{A}\) has \(3\times 6=18\) independent components. The constraint \((\bb{A})_{iik}=0\) imposes three independent conditions, and therefore
\begin{equation}\label{eq:dim.A}
\dim \cl{A} = 18-3 = 15.
\end{equation}

We now derive the three components in \eqref{eq:A7.mode.lemma}, \eqref{eq:A3.mode.lemma}, and \eqref{eq:A5.mode.lemma}. First, the only vector naturally induced by \(\bb{A} \in \cl{A}\) is
\begin{equation}\label{eq:trace.vector.lemma}
a_i=(\bb{A})_{imm}.
\end{equation}
The most general isotropic tensor depending linearly on \(a_i\) and satisfying the minor symmetry \(j\leftrightarrow k\) has the form
\begin{equation}\label{eq:vector.generated.ansatz}
(\bb{A}^{(3)})_{ijk} = \alpha\, \delta_{jk}a_i + \beta \, (\delta_{ij}a_k+\delta_{ik}a_j).
\end{equation}
Imposing the traceless condition on the first two indices \((\bb{A}^{(3)})_{iik}=0\) results in
\begin{equation}\label{eq:vector.generated.constraint.1}
\alpha + 4\beta = 0,
\end{equation}
whereas
\begin{equation}\label{eq:vector.generated.trace}
(\bb{A}^{(3)})_{imm}=(3\alpha+2\beta)a_i,
\end{equation}
and requiring \((\bb{A}^{(3)})_{imm}=a_i\) yields
\begin{equation}\label{eq:vector.generated.constraint.2}
3\alpha + 2\beta = 1.
\end{equation}
Solving \eqref{eq:vector.generated.constraint.1} and \eqref{eq:vector.generated.constraint.2} yields
\begin{equation}\label{eq:vector.generated.coefficients}
\alpha = \dfrac{2}{5} \texand \beta = -\dfrac{1}{10},
\end{equation}
which proves \eqref{eq:A3.mode.lemma}.

Next, consider the symmetrization
\begin{equation}\label{eq:full.symmetrization}
(\sym_3 \bb{A})_{ijk} \coloneqq \fr{1}{3} ((\bb{A})_{ijk} + (\bb{A})_{jik} + (\bb{A})_{kij}).
\end{equation}
Because \((\bb{A})_{ijk} = (\bb{A})_{ikj}\), the tensor \(\sym_3 \bb{A}\) is fully symmetric in \(i,j,k\) and its trace is
\begin{align}\label{eq:sym.trace}
(\sym_3 \bb{A})_{imm} &= \fr{1}{3} ((\bb{A})_{imm} + (\bb{A})_{mim} + (\bb{A})_{mmi}) \nonumber\\[4pt]
&= \fr{1}{3} a_i,
\end{align}
where we have used the symmetry on the last two indices along with the traceless nature \((\bb{A})_{mim} = (\bb{A})_{mmi} = 0\). The only fully symmetric isotropic tensor induced by \(a_i\) has the form
\begin{equation}\label{eq:fully.symmetric.isotropic.vector}
\gamma (\delta_{ij}a_k + \delta_{ik}a_j + \delta_{jk}a_i),
\end{equation}
for which its trace is \(5\gamma a_i\). To annihilate the trace  of \(\sym_3 \bb{A}\) in \eqref{eq:sym.trace}, we impose \(5\gamma = \fr{1}{3}\), that is, \(\gamma = \fr{1}{15}\). Therefore, the fully symmetric traceless part \(\bb{A}^{(7)}\) is given by
\begin{equation}\label{eq:A7.derived}
(\bb{A}^{(7)})_{ijk} = \fr{1}{3} ((\bb{A})_{ijk}+(\bb{A})_{jik}+(\bb{A})_{kij}) - \fr{1}{15} (\delta_{ij}a_k + \delta_{ik}a_j + \delta_{jk}a_i).
\end{equation}
This proves \eqref{eq:A7.mode.lemma}.

Finally, the remainder
\begin{equation}\label{eq:A5.remainder}
(\bb{A}^{(5)})_{ijk} = (\bb{A})_{ijk}-(\bb{A}^{(7)})_{ijk}-(\bb{A}^{(3)})_{ijk}
\end{equation}
belongs to \(\cl{A}\) and satisfies \eqref{eq:A5.props.lemma}. Expanding \eqref{eq:A5.remainder} using \eqref{eq:A3.mode.lemma} and \eqref{eq:A7.mode.lemma} produces \eqref{eq:A5.mode.explicit.lemma}. Identities \eqref{eq:A7.props.lemma}, \eqref{eq:A3.props.lemma}, and \eqref{eq:A5.props.lemma} are proved by direct substitution of \eqref{eq:A7.mode.lemma}, \eqref{eq:A3.mode.lemma}, and \eqref{eq:A5.mode.lemma}.

Since \((\bb{A}^{(7)})_{ijk}\) is fully symmetric and traceless, while \((\bb{A}^{(3)})_{ijk}\) is given by \eqref{eq:A3.mode.lemma}, we have
\begin{align}\label{eq:orthogonality.73.proof}
\bb{A}^{(7)} \threevdots \bb{A}^{(3)} &= (\bb{A}^{(7)})_{ijk}(\bb{A}^{(3)})_{ijk} \nonumber\\[4pt]
&= (\bb{A}^{(7)})_{ijk} (\fr{2}{5} \delta_{jk}a_i - \fr{1}{10}(\delta_{ij}a_k + \delta_{ik}a_j)) \nonumber\\[4pt]
&= \fr{2}{5} (\bb{A}^{(7)})_{ijj}a_i - \fr{1}{10} (\bb{A}^{(7)})_{iik}a_k - \fr{1}{10} (\bb{A}^{(7)})_{iki}a_k = 0,
\end{align}
because all traces of \(\bb{A}^{(7)}\) vanish identically. Moreover, since \(\bb{A}^{(7)}\) is fully symmetric and \(\bb{A}^{(5)}\) satisfies the cyclic identity \eqref{eq:A5.props.lemma}\(_3\), we have that
\begin{equation}
\bb{A}^{(7)}\threevdots\bb{A}^{(5)} = \fr{1}{3} (\bb{A}^{(7)})_{ijk} ((\bb{A}^{(5)})_{ijk} + (\bb{A}^{(5)})_{jik} + (\bb{A}^{(5)})_{kij}) =0.
\end{equation}
Similarly, using the trace-free properties of \(\bb{A}^{(5)}\),
\begin{equation}
\bb{A}^{(5)}\threevdots\bb{A}^{(3)} = \fr{2}{5}(\bb{A}^{(5)})_{ijj}a_i - \fr{1}{10}(\bb{A}^{(5)})_{iik}a_k - \fr{1}{10}(\bb{A}^{(5)})_{iki}a_k =0,
\end{equation}
proving the orthogonality condition \eqref{eq:orthogonality.modes.lemma}.

The space \(\cl{H}^7\) consists of fully symmetric traceless third-order tensors in three dimensions and therefore has dimension \(7\). The space \(\cl{H}^3\) is parametrized uniquely by the vector \(a_i\), and hence has dimension \(3\). Since \(\cl{H}^5\) is the orthogonal complement of \(\cl{H}^7\oplus\cl{H}^3\) in \(\cl{A}\), it follows from \eqref{eq:dim.A}, that \(\cl{H}^5\) has dimension \(15-7-3=5\). Since the three subspaces are mutually orthogonal and their dimensions add,
\begin{equation}\label{eq:dimensions.add}
\dim \cl{H}^7 + \dim \cl{H}^5 + \dim \cl{H}^3 = 7 + 5 + 3 = 15 = \dim \cl{A},
\end{equation}
the decomposition \eqref{eq:A.decomposition} follows. Uniqueness is immediate from orthogonality.
The projection formulas \eqref{eq:A7.mode.lemma}--\eqref{eq:A5.mode.lemma} are \(\mathrm{SO}(3)\)-equivariant. Their ranges are the standard harmonic tensor representations of orders \(3\), \(2\), and \(1\), respectively, and are therefore irreducible under \(\mathrm{SO}(3)\); see Espath \cite[Chapter 5]{Esp23}.
\end{proof}

\begin{lem}[Positivity conditions for the sixth-order dissipative operator]\label{lem:sharp.positivity.C6}
Let \(\cl{A}\) be the space defined in \eqref{eq:E.space}, and let \(\pmb{\bb{C}} \colon \cl{A} \to \cl{A}\) be the isotropic sixth-order tensor defined by
\begin{align}\label{eq:C6.action.traceless.theorem}
(\pmb{\bb{C}}\bb{A})_{ijk} ={}& 2\ellcg{6}{3} ((\bb{A})_{kij}+(\bb{A})_{jik} - \fr{1}{4} (\delta_{ij}(\bb{A})_{kmm} + \delta_{ik}(\bb{A})_{jmm}))\nonumber\\[4pt]
&+ 2\ellcg{6}{4}(\bb{A})_{ijk} + \ellcg{6}{5} (\delta_{jk}(\bb{A})_{imm} - \fr{1}{4} (\delta_{ij}(\bb{A})_{kmm} + \delta_{ik}(\bb{A})_{jmm})).
\end{align}
Then, the quadratic form associated with \(\pmb{\bb{C}}\) satisfies
\begin{equation}\label{eq:Q.C6.theorem}
\cl{Q}_{\pmb{\bb{C}}}(\bb{A}) \coloneqq \bb{A} \threevdots \pmb{\bb{C}} \bb{A} \ge 0 \forevery \bb{A} \in \cl{A},
\end{equation}
if and only if
\begin{equation}\label{eq:C6.sharp.conditions}
2\ellcg{6}{3}+\ellcg{6}{4}\ge 0, \qquad \ellcg{6}{4}-\ellcg{6}{3}\ge 0, \texand -2\ellcg{6}{3} + 4\ellcg{6}{4} + 5\ellcg{6}{5}\ge 0.
\end{equation}
Equivalently, if \(\bb{A}=\bb{A}^{(7)}+\bb{A}^{(5)}+\bb{A}^{(3)}\) is given by the harmonic decomposition from Lemma \ref{lem:harmonic.decomp.A}, then
\begin{equation}\label{eq:Q.harmonic.split.theorem}
\bb{A} \threevdots \pmb{\bb{C}} \bb{A} = \lambda_7 |\bb{A}^{(7)}|^2 +\lambda_5 |\bb{A}^{(5)}|^2 +\lambda_3 |\bb{A}^{(3)}|^2,
\end{equation}
and
\begin{equation}
\left\{
\begin{aligned}
\lambda_7 \coloneqq {}&4\ellcg{6}{3}+2\ellcg{6}{4}, \\[4pt] \lambda_5 \coloneqq {}&-2\ellcg{6}{3}+2\ellcg{6}{4}, \\[4pt] \lambda_3 \coloneqq {}&-\ellcg{6}{3}+2\ellcg{6}{4}+\fr52\ellcg{6}{5}.
\end{aligned}
\right.
\end{equation}
\end{lem}
\begin{proof}
From \eqref{eq:6.tensor.isotropic.inner.product.traceless}, the quadratic form associated with \(\pmb{\bb{C}}\) is
\begin{equation}\label{eq:Q.C6.expanded}
(\bb{A})_{ijk}(\pmb{\bb{C}}\bb{A})_{ijk} = 2\ellcg{6}{3}\big((\bb{A})_{ijk}(\bb{A})_{kij}+(\bb{A})_{ijk}(\bb{A})_{jik}\big) + 2\ellcg{6}{4}(\bb{A})_{ijk}(\bb{A})_{ijk} + \ellcg{6}{5}(\bb{A})_{ikk}(\bb{A})_{imm}.
\end{equation}
By Lemma \ref{lem:harmonic.decomp.A}, we use the orthogonal decomposition of \(\bb{A}\) into \(\mathrm{SO}(3)\)-irreducible components. Direct substitution of the defining identities \eqref{eq:A7.props.lemma}--\eqref{eq:A3.props.lemma} into \eqref{eq:C6.action.traceless.theorem} gives
\begin{equation}\label{eq:C.mode.action}
\pmb{\bb{C}}\bb{A}^{(7)} = \lambda_7\bb{A}^{(7)}, \qquad \pmb{\bb{C}}\bb{A}^{(5)} = \lambda_5\bb{A}^{(5)}, \texand \pmb{\bb{C}}\bb{A}^{(3)} = \lambda_3\bb{A}^{(3)}.
\end{equation}
Thus, using the harmonic decomposition \eqref{eq:A.harmonic.decomposition}, orthogonality \eqref{eq:orthogonality.modes.lemma} and \eqref{eq:C.mode.action}, we have that
\begin{align}\label{eq:Q.mode.diagonal}
\bb{A}\threevdots \pmb{\bb{C}}\bb{A} ={}& (\bb{A}^{(7)} + \bb{A}^{(5)} + \bb{A}^{(3)}) \threevdots (\lambda_7\bb{A}^{(7)} + \lambda_5\bb{A}^{(5)} + \lambda_3\bb{A}^{(3)}) \nonumber\\[4pt]
&= \lambda_7|\bb{A}^{(7)}|^2 + \lambda_5|\bb{A}^{(5)}|^2 + \lambda_3|\bb{A}^{(3)}|^2.
\end{align}
It remains only to determine the three scalars.

For \(\bb{A}=\bb{A}^{(7)}\), with \eqref{eq:A7.props.lemma} from Lemma \ref{lem:harmonic.decomp.A}, we have
\begin{equation}\label{eq:H7.identities}
(\bb{A})_{kij} = (\bb{A})_{ijk}, \qquad (\bb{A})_{jik} = (\bb{A})_{ijk}, \texand (\bb{A})_{ikk} = 0.
\end{equation}
Substituting \eqref{eq:H7.identities} into \eqref{eq:Q.C6.expanded}, we find
\begin{align}
\bb{A}^{(7)} \threevdots \pmb{\bb{C}}\bb{A}^{(7)} &= 2\ellcg{6}{3}((\bb{A}^{(7)})_{ijk}(\bb{A}^{(7)})_{ijk} + (\bb{A}^{(7)})_{ijk}(\bb{A}^{(7)})_{ijk}) + 2\ellcg{6}{4}(\bb{A}^{(7)})_{ijk}(\bb{A}^{(7)})_{ijk}\nonumber\\[4pt]
&= (4\ellcg{6}{3} + 2\ellcg{6}{4})|\bb{A}^{(7)}|^2.
\end{align}
Thus,
\begin{equation}\label{eq:lambda7.derived}
\lambda_7 = 4\ellcg{6}{3} + 2\ellcg{6}{4}.
\end{equation}

For \(\bb{A}=\bb{A}^{(5)}\), with \eqref{eq:A5.props.lemma} from Lemma \ref{lem:harmonic.decomp.A}, we have that
\begin{equation}\label{eq:H5.identities}
(\bb{A})_{ijk}+(\bb{A})_{jik}+(\bb{A})_{kij} = 0 \texand (\bb{A})_{ikk} = 0.
\end{equation}
Substituting \eqref{eq:H5.identities}\(_1\) into \eqref{eq:Q.C6.expanded}, we obtain
\begin{align}
\bb{A}^{(5)} \threevdots \pmb{\bb{C}}\bb{A}^{(5)} &= 2\ellcg{6}{3}(\bb{A}^{(5)})_{ijk} ((\bb{A}^{(5)})_{kij} + (\bb{A}^{(5)})_{jik}) + 2\ellcg{6}{4}(\bb{A}^{(5)})_{ijk} (\bb{A}^{(5)})_{ijk} \nonumber\\[4pt]
&= (-2\ellcg{6}{3}+2\ellcg{6}{4})|\bb{A}^{(5)}|^2.
\end{align}
Thus,
\begin{equation}\label{eq:lambda5.derived}
\lambda_5 = -2\ellcg{6}{3}+2\ellcg{6}{4}.
\end{equation}

For \(\bb{A}=\bb{A}^{(3)}\), set \(a_i \coloneqq (\bb{A}^{(3)})_{imm}\). By Lemma \ref{lem:harmonic.decomp.A},
\begin{equation}\label{eq:A3.mode.theorem}
(\bb{A}^{(3)})_{ijk} = \fr{2}{5} \delta_{jk}a_i - \fr{1}{10} (\delta_{ij}a_k+\delta_{ik}a_j).
\end{equation}
A direct contraction yields
\begin{equation}\label{eq:A3.norm.theorem}
(\bb{A}^{(3)})_{ijk}(\bb{A}^{(3)})_{ijk} = \fr{2}{5} a_i a_i,
\end{equation}
and
\begin{equation}\label{eq:A3.cross.theorem}
\left\{
\begin{aligned}
(\bb{A}^{(3)})_{ijk}(\bb{A}^{(3)})_{kij} &= -\fr{1}{10} a_i a_i, \\[4pt] (\bb{A}^{(3)})_{ijk}(\bb{A}^{(3)})_{jik} &= -\fr{1}{10} a_i a_i, \\[4pt] (\bb{A}^{(3)})_{ikk}(\bb{A}^{(3)})_{imm} &= a_i a_i.
\end{aligned}
\right.
\end{equation}
Substituting \eqref{eq:A3.cross.theorem} into \eqref{eq:Q.C6.expanded}, we obtain
\begin{align}\label{eq:Q.H3.intermediate}
\bb{A}^{(3)} \threevdots \pmb{\bb{C}}\bb{A}^{(3)} &= 2\ellcg{6}{3} (-\fr{1}{10}a_i a_i - \fr{1}{10}a_i a_i) + 2\ellcg{6}{4}\fr{2}{5}a_i a_i + \ellcg{6}{5}a_i a_i \nonumber\\[4pt]
&=(-\fr{2}{5}\ellcg{6}{3} + \fr{4}{5}\ellcg{6}{4} + \ellcg{6}{5}) a_i a_i.
\end{align}
Using \eqref{eq:A3.norm.theorem}, we arrive at
\begin{equation}\label{eq:Q.H3.final}
\bb{A}^{(3)} \threevdots \pmb{\bb{C}}\bb{A}^{(3)} = (-\ellcg{6}{3} + 2\ellcg{6}{4} + \fr{5}{2}\ellcg{6}{5})|\bb{A}^{(3)}|^2.
\end{equation}
Thus,
\begin{equation}\label{eq:lambda3.derived}
\lambda_3 = -\ellcg{6}{3} + 2\ellcg{6}{4} + \fr{5}{2}\ellcg{6}{5}.
\end{equation}

Finally, we have proved that \eqref{eq:Q.C6.theorem} holds if and only if
\begin{equation}\label{eq:lambda.nonnegative}
\lambda_7 \ge 0, \qquad \lambda_5 \ge 0, \texand \lambda_3 \ge 0.
\end{equation}
In terms of the coefficients \(\ellcg{6}{3}\), \(\ellcg{6}{4}\), and \(\ellcg{6}{5}\), \eqref{eq:lambda.nonnegative} is equivalent to
\begin{equation}\label{eq:C6.conditions.intermediate}
4\ellcg{6}{3} + 2\ellcg{6}{4}\ge 0, \qquad -2\ellcg{6}{3} + 2\ellcg{6}{4}\ge 0, \texand -\ellcg{6}{3} + 2\ellcg{6}{4} + \fr{5}{2}\ellcg{6}{5}\ge 0.
\end{equation}
\end{proof}

Finally, invoking the canonical free-energy imbalance \eqref{eq:pointwise.dissipation.final.reduced.equivalent}, along with the constitutive restrictions \eqref{eq:constraint}, the choices \eqref{eq:Xi.dis}, and the spectral decomposition \eqref{eq:Q.harmonic.split.theorem} under the positivity conditions \eqref{eq:C6.sharp.conditions}, the dissipation inequality reduces to
\begin{equation}\label{eq:dissipation.inequality.final}
2 \mu |\br{D}|^2 + \lambda_7 |(\Grad\br{L})^{(7)}|^2 + \lambda_5 |(\Grad\br{L})^{(5)}|^2 + \lambda_3 |(\Grad\br{L})^{(3)}|^2 + \beta |\mathfrak{D}_{\!\alpha}\br{J}|^2 \ge 0.
\end{equation}

\begin{rmk}[Interpretation of the dissipative modes]\label{rmk:C6.modes}
The three inequalities in \eqref{eq:C6.sharp.conditions} correspond to the three orthogonal harmonic modes in the decomposition of \(\bb{A}\).

When \(\bb{A}=\Grad\br{L}\), the mode \(\bb{A}^{(7)}\in\cl{H}^7\) is the fully symmetric traceless part of the gradient of the velocity gradient and represents a pure extensional-curvature mode. The mode \(\bb{A}^{(5)}\in\cl{H}^5\) is the trace-free mixed-symmetry part and represents a distortional-curvature or shear-gradient mode. The mode \(\bb{A}^{(3)}\in\cl{H}^3\) is generated by the trace vector \(a_i=(\bb{A})_{imm}\), such that
\begin{equation}\label{eq:trace.vector.Laplacian}
a_i = \partial_m (\br{L})_{im} = \partial_m\partial_m v_i = \Delta v_i,
\end{equation}
so that \(\bb{A}^{(3)}\) is a diffusion-like or Laplacian-curvature mode. Here, \(\Delta\) represents the Laplacian operator. Accordingly, \(\lambda_7\), \(\lambda_5\), and \(\lambda_3\) are three independent isotropic hyperviscosity moduli associated with these three modes.
\end{rmk}

\section{Final governing equations, constitutive relations, and boundary conditions}\label{sc:final.model}

In this section, we summarize the final system of field equations resulting from the preceding kinematical, balance, and thermodynamic considerations. Again, we restrict attention to incompressible and isothermal processes.

\subsection{Field equations}

The bulk unknown fields are the velocity \(\vel\), the bulk pressure \(\pi\), and the internal variable \(\br{J}\), where
\begin{equation}\label{eq:bulk.effective.pressure}
\pi \coloneqq p-\Div\bs{p}.
\end{equation}
Indeed, the scalar pressure \(p\) and the hyperpressure \(\bs{p}\) enter the bulk momentum balance only through \(\pi\). Their decomposition is invariant under the bulk gauge transformation
\begin{equation}\label{eq:pressure.hyperpressure.gauge}
(p,\bs{p})\longmapsto(p+\Div\bs{q},\bs{p}+\bs{q})
\end{equation}
for every sufficiently smooth vector field \(\bs{q}\). Thus, \(p\) and \(\bs{p}\) are not independent bulk unknowns. The closed bulk system consists of the mass balance, linear momentum balance, and microforce balance:
\begin{equation}\label{eq:field.equations}
\left\{
\begin{aligned}
& \Div \vel = 0, \\[4pt] &
\varrho \dot{\vel} = \Div (\br{S} - \Div\bb{S})-\Grad\pi+\varrho \bs{b}, \\[4pt]
& \bs{\Xi} + \bs{\Upsilon} + \Div \bb{X} = \bs{0},
\end{aligned}
\right.
\end{equation}
Since \(\br{J}\in\mathrm{Sym}\), the microforce balance is an equality in \(\mathrm{Sym}\) and supplies six independent scalar equations. Hence, the bulk system comprises ten scalar equations for the ten scalar components of \(\vel\), \(\pi\), and \(\br{J}\). Here,
\begin{align}
\br{H}_{\alpha} \coloneqq{}& \alpha (\br{J}\partial_{\br{J}}\psi + \partial_{\br{J}}\psi\br{J}) \nonumber\\[4pt]
&+ \sym\big(\alpha (\partial_{\Grad\br{J}}\psi \twovcR \Grad\br{J} + \partial_{\Grad\br{J}}\psi \twovcM \Grad\br{J}) - \Grad\br{J} \twovc \partial_{\Grad\br{J}}\psi\big),
\end{align}
and
\begin{equation}
\bb{H}_{\alpha} \coloneqq \fr{1+\alpha}{2}\bb{H}+\fr{\alpha-1}{2}\bb{H}^{\sperp}, \with \bb{H} \coloneqq ((\partial_{\Grad\br{J}}\psi)^{\trans}\br{J})^{\!\trans}+(\br{J}\partial_{\Grad\br{J}}\psi)^{\sperp},
\end{equation}
and
\begin{equation}
\left\{
\begin{aligned}
\br{S} &= 2\mu \br{D} + \bb{P}_0[\br{H}_{\alpha}], \\[4pt] \bb{S} &= \lambda_7 (\Grad\br{L})^{(7)} + \lambda_5 (\Grad\br{L})^{(5)} + \lambda_3 (\Grad\br{L})^{(3)} + \bb{P}_{\!\!\scriptscriptstyle\cl{A}}[\sym_{23}\bb{H}_{\alpha}], \\[4pt] \bs{\Xi} &= -\beta \, \mathfrak{D}_{\!\alpha}\br{J} - \partial_{\br{J}}\psi, \\[4pt] \bb{X} &= \partial_{\Grad\br{J}}\psi,
\end{aligned}
\right.
\end{equation}
The projection in the constitutive relation for \(\br{S}\) fixes its deviatoric representative. Since the bulk pressure is indeterminate, the spherical part of \(\br{H}_{\alpha}\) can be absorbed into an effective pressure; equivalently,
\begin{equation}
-\pi\bs{1}+\bb{P}_0[\br{H}_{\alpha}] = -\pi_{\mathrm{eff}}\bs{1}+\br{H}_{\alpha}, \qquad \pi_{\mathrm{eff}} = \pi+\fr{1}{3}\tr\br{H}_{\alpha}.
\end{equation}
The free energy is given by
\begin{equation}
\psi = \hat{\psi}(\br{J},\Grad\br{J}),
\end{equation}
and
\begin{equation}
\left\{
\begin{aligned}
\mu &\ge 0, \qquad& \text{shear viscosity}, \\ \beta &\ge 0, \qquad& \text{internal-variable drag coefficient}, \\ \lambda_7 &\ge 0, \qquad& \text{first hyperviscosity modulus}, \\ \lambda_5 &\ge 0, \qquad& \text{second hyperviscosity modulus}, \\ \lambda_3 &\ge 0, \qquad& \text{third hyperviscosity modulus}.
\end{aligned}
\right.
\end{equation}
The parameter \(\alpha \in \{0,1\}\) selects the transport structure: \(\alpha = 0\) for corotational advection or \(\alpha = 1\) for advection associated with the upper-convected rate. Note that the microforce balance may be rewritten as the following Allen--Cahn-type evolution equation
\begin{equation}
\beta \, \mathfrak{D}_{\!\alpha}\br{J} = - \partial_{\br{J}}\psi + \Div(\partial_{\Grad\br{J}}\psi) + \bs{\Upsilon}.
\end{equation}
The reader is referred to the work by Espath \& Calo \cite{Esp21} and Espath et al. \cite{Esp20} for different treatments of gradient theories of phase transition.

\subsection{Boundary conditions}

Throughout this subsection, an overbar denotes prescribed boundary data. To express the natural tractions, we retain \(p\) and \(\bs{p}\) as Lagrange-multiplier representatives satisfying \eqref{eq:bulk.effective.pressure}; a representative of their bulk gauge class is chosen consistently with the boundary data. The natural boundary conditions are expressed in terms of the associated tractions. Thus, on the traction part of the boundary, one may prescribe the surface traction
\begin{equation}\label{eq:final.natural.t}
\bsts = (\br{T} - \Div \bb{T})\bs{n} - \Divs ((\bb{T}\bs{n})\br{P}) = \overline{\bsts}, \qquad\text{on }\partial \prt_t^{\bsts},
\end{equation}
and the hypertraction
\begin{equation}\label{eq:final.natural.h}
\bshs = (\bb{T}\bs{n})\bs{n} = \overline{\bshs}, \qquad\text{on }\partial \prt_t^{\bshs}.
\end{equation}
If the boundary contains edges, one also has the natural edge traction
\begin{equation}\label{eq:final.natural.edge}
\bstc = \surp{(\bb{T}\bs{n})\bs{\nu}} = \overline{\bstc}, \qquad\text{on }\partial^2 \prt_t^{\bstc}.
\end{equation}
As for the microtraction, one may prescribe
\begin{equation}\label{eq:final.natural.micro}
\bsxs = \bb{X}\bs{n} = \overline{\bsxs}, \qquad\text{on }\partial \prt_t^{\bsxs}.
\end{equation}

The essential boundary conditions associated with the motion are
\begin{equation}\label{eq:final.essential.v}
\vel = \overline{\vel} \qquad\text{on } \partial\prt_t^{\vel},
\end{equation}
and
\begin{equation}\label{eq:final.essential.dv}
\partial_n \vel = \overline{\bs{g}} \qquad\text{on } \partial\prt_t^{\partial_n \vel}.
\end{equation}
The essential boundary condition for the internal variable is
\begin{equation}\label{eq:final.essential.J}
\br{J} = \overline{\br{J}}, \qquad\text{on } \partial\prt_t^{\br{J}}.
\end{equation}

The boundary is therefore partitioned as
\begin{equation}
\partial\prt_t = \partial\prt_t^{\,\vel}\cup \partial\prt_t^{\,\bsts} = \partial\prt_t^{\,\partial_n \vel}\cup \partial\prt_t^{\,\bshs} = \partial\prt_t^{\br{J}}\cup \partial\prt_t^{\bsxs},
\end{equation}
with the corresponding pairs of sets mutually disjoint.

\section{Gradient Oldroyd-B coupled with Landau--de Gennes theories}\label{sc:gradient.OB.LD}

Define the cone of positive-definite symmetric tensors by
\begin{equation}\label{eq:Sym.plus.space}
\mathrm{Sym}^+ \coloneqq \{\br{A}\in\mathrm{Sym}\colon \bs{a}\cdot\br{A}\bs{a}>0\ \text{for every}\ \bs{a}\in\cl{V}\setminus\{\bs{0}\}\}.
\end{equation}
Here, we consider the conformation tensor \(\br{A} \in \mathrm{Sym}^+\), a measure of deformation, as our first internal variable. This internal variable is assumed to advect according to the transport structure associated with the upper-convected rate, thus its frame-indifferent rate is \(\mathfrak{D}_{\!1} \br{A}\). Additionally, we consider a second internal variable \(\br{Y} \in \mathrm{Sym}_0\), with \(\mathrm{Sym}_0\) defined in \eqref{eq:Sym0.space}, representing the orientation tensor in the sense of the Landau--de Gennes theory or as a fabric tensor describing anisotropic microstructure, depending on the physical context. This internal variable is assumed to advect corotationally with the material, thus its frame-indifferent rate is \(\mathfrak{D}_{\!0} \br{Y}\). This choice reflects the distinct physical roles of the internal variables. We restrict attention to sufficiently smooth solutions for which \(\br{A}(\bs{y},t)\in\mathrm{Sym}^+\) throughout the evolution. Let \(\bs{\Upsilon}_1\in\mathrm{Sym}\) and \(\bs{\Upsilon}_0\in\mathrm{Sym}\) denote the external microforces associated with \(\br{A}\) and \(\br{Y}\), respectively; only \(\bb{P}_0[\bs{\Upsilon}_0]\) is power-active in the constrained orientational evolution.

As for the orientation tensor, it is customary to interpret it in the uniaxial limit as arising from a preferred direction \(\bs{n}\), that is,
\begin{equation}
\br{Y} \propto \bs{n} \otimes \bs{n} - \fr{1}{3} \bs{1},
\end{equation}
where \(\bs{n} \in T_{\bs{y}} \prt_t\) is a unit vector representing the preferred orientation in the Landau--de Gennes theory. However, in the context of fabric tensors,
\begin{equation}
\br{Y} \propto \dfrac{1}{N} \sum_i \bs{n}_i \otimes \bs{n}_i - \fr{1}{3} \bs{1}
\end{equation}
where \(N\) is the number of fibers and \(\bs{n}_i\in T_{\bs{y}}\prt_t\) is a unit vector directed along the \(i\)-th fiber. In the present formulation \(\br{Y}\) is treated as a general symmetric tensor satisfying \(\tr\br{Y}=0\).

Regardless, the free energy density response function is defined as
\begin{equation}\label{eq:free.energy.OB.LD}
\psi \coloneqq \hat{\psi}(\br{A},\br{Y},\Grad\br{A},\Grad\br{Y}),
\end{equation}
such that\footnote{Additional coupling terms between the conformation tensor \(\br{A}\) and the orientation tensor \(\br{Y}\) may be incorporated to account for interactions between elastic deformation and anisotropic microstructure. Such terms must be constructed from scalar invariants of \(\br{A}\) and \(\br{Y}\) in order to preserve isotropy and frame indifference. For instance, an admissible coupling is given by
\begin{equation}
\sigma \, \tr(\br{A}\br{Y}^2),
\end{equation}
where \(\sigma > 0\) is an additional material parameter.}
\begin{equation}\label{eq:free.energy.OB.LD.explicit}
\psi = \dfrac{G}{2} (\tr \br{A}-\ln\det\br{A}-3) + \dfrac{a}{2}\,\tr(\br{Y}^2) - \dfrac{b}{3}\,\tr(\br{Y}^3) + \dfrac{c}{4}\,(\tr(\br{Y}^2))^2 + \dfrac{\gamma_1}{2}|\Grad\br{A}|^2 + \dfrac{\gamma_0}{2}|\Grad\br{Y}|^2.
\end{equation}
Here, \(G>0\) is the polymer elastic modulus. The Landau--de Gennes bulk coefficients are \(a\in\bb{R}\), which may depend on temperature and change sign, \(b>0\), which controls the cubic contribution associated with the first-order nematic transition, and \(c>0\), which ensures coercivity of the bulk orientational energy. The coefficients \(\gamma_1\ge 0\) and \(\gamma_0\ge 0\) are the conformation-gradient and orientational-gradient moduli, respectively; they are strictly positive when the corresponding gradient regularization is retained.

With respect to the orientational variable, the free energy is first regarded as a function on the ambient space \(\mathrm{Sym}\) and is then restricted to \(\mathrm{Sym}_0\). Its ambient variational derivatives are not, in general, tangent to \(\mathrm{Sym}_0\). Therefore, for the constrained orientational variable \(\br{Y} \in \mathrm{Sym}_0\), the microforce balance is projected onto \(\mathrm{Sym}_0\). The derivatives with respect to \(\br{Y}\) and \(\Grad\br{Y}\) are first taken in the ambient space, and the corresponding variational driving force is then projected onto \(\mathrm{Sym}_0\). This ensures that the evolution preserves symmetry and tracelessness.
\begin{prop}[Projection onto \(\mathrm{Sym}_0\) and dissipative consistency]\label{prop:P0}
Let \(\br{Z} \in \bb{R}^{3 \times 3}\), and recall the projection \(\bb{P}_0 \colon \bb{R}^{3 \times 3} \to \mathrm{Sym}_0\) defined in \eqref{eq:P0.def}. Then
\begin{enumerate}[label=\textit{(\roman*)}]
\item \(\bb{P}_0\) is the orthogonal projection onto \(\mathrm{Sym}_0\) with respect to the Euclidean inner product.
\item If \(\br{Y}(0)\in \mathrm{Sym}_0\) and its evolution is given by
\begin{equation}\label{eq:P0.evolution}
\mathfrak{D}_{\!0}\br{Y} = \bb{P}_0[\br{Z}],
\end{equation}
for some \(\br{Z}\), then \(\br{Y}(t)\in \mathrm{Sym}_0\) for all time.
\item If the dissipative microforce is chosen as
\begin{equation}\label{eq:P0.Xi}
\bs{\Xi}^{\mathrm{dis}}_0 = -\beta_0 \bb{P}_0[\mathfrak{D}_{\!0}\br{Y}], \qquad \beta_0 > 0,
\end{equation}
then the contribution to the free-energy imbalance satisfies
\begin{equation}\label{eq:P0.dissipation}
\bs{\Xi}^{\mathrm{dis}}_0 \twovdots \mathfrak{D}_{\!0}\br{Y} = -\beta_0 |\bb{P}_0[\mathfrak{D}_{\!0}\br{Y}]|^2 \le 0.
\end{equation}
\end{enumerate}
\end{prop}
\begin{proof}
\indent\textit{(i)} The operator \(\bb{P}_0\) is linear and idempotent, that is,
\begin{equation}\label{eq:P0.idempotent}
\bb{P}_0^2 = \bb{P}_0.
\end{equation}
Moreover, for any \(\br{Z}\),
\begin{equation}\label{eq:P0.decomposition}
\br{Z} - \bb{P}_0[\br{Z}] = \fr{1}{3} \tr(\br{Z})\bs{1} + \skw(\br{Z}),
\end{equation}
is orthogonal to \(\mathrm{Sym}_0\). Hence \(\bb{P}_0\) is the orthogonal projection onto \(\mathrm{Sym}_0\).

\medskip
\indent\textit{(ii)} If \(\br{Y}(0)\in \mathrm{Sym}_0\) and
\begin{equation}\label{eq:P0.evolution.proof}
\mathfrak{D}_{\!0}\br{Y}=\bb{P}_0[\br{Z}],
\end{equation}
then \(\mathfrak{D}_{\!0}\br{Y}\in \mathrm{Sym}_0\). For the corotational rate \(\mathfrak{D}_{\!0}\), symmetry and trace are preserved, hence \(\br{Y}(t)\in \mathrm{Sym}_0\) for all time.

\medskip
\indent\textit{(iii)} Using the orthogonality of \(\bb{P}_0\),
\begin{equation}
- \beta_0 \bb{P}_0[\mathfrak{D}_{\!0}\br{Y}] \twovdots \mathfrak{D}_{\!0}\br{Y} = - \beta_0 \bb{P}_0[\mathfrak{D}_{\!0}\br{Y}] \twovdots \bb{P}_0[\mathfrak{D}_{\!0}\br{Y}] = - \beta_0 |\bb{P}_0[\mathfrak{D}_{\!0}\br{Y}]|^2 \le 0.
\end{equation}
\end{proof}

By additivity of the free-energy rate, the energetic stresses associated with the two internal variables are obtained by summing their individual contributions. For the second-order contributions, we write
\begin{align}
\br{H}_{1}[\br{A}] \coloneqq{}& \br{A}\partial_{\br{A}}\psi + \partial_{\br{A}}\psi\br{A} \nonumber\\[4pt]
&+ \sym\big(\partial_{\Grad\br{A}}\psi \twovcR \Grad\br{A} + \partial_{\Grad\br{A}}\psi \twovcM \Grad\br{A} - \Grad\br{A} \twovc \partial_{\Grad\br{A}}\psi\big), \\[4pt]
\br{H}_{0}[\br{Y}] \coloneqq{}& -\sym\big(\Grad\br{Y} \twovc \partial_{\Grad\br{Y}}\psi\big).
\end{align}
The corresponding third-order building blocks are
\begin{equation}
\left\{
\begin{aligned}
\bb{H}[\br{A}] \coloneqq{}& ((\partial_{\Grad\br{A}}\psi)^{\trans}\br{A})^{\!\trans}+(\br{A}\partial_{\Grad\br{A}}\psi)^{\sperp}, \\[4pt]
\bb{H}[\br{Y}] \coloneqq{}& ((\partial_{\Grad\br{Y}}\psi)^{\trans}\br{Y})^{\!\trans}+(\br{Y}\partial_{\Grad\br{Y}}\psi)^{\sperp}.
\end{aligned}
\right.
\end{equation}
Writing \(\skw_{\!12}\bb{C} \coloneqq \fr{1}{2}(\bb{C}-\bb{C}^{\sperp})\) for skew-symmetrization with respect to the first two indices, the transport-dependent energetic hyperstresses defined in Theorem \ref{thm:canonical.free.energy.imbalance} specialize to \(\bb{H}_{1}[\br{A}]=\bb{H}[\br{A}]\) and \(\bb{H}_{0}[\br{Y}]=\skw_{\!12}\bb{H}[\br{Y}]\).

Let
\begin{equation}
\psi_{\mathrm{OB}}(\br{A}) \coloneqq \dfrac{G}{2}(\tr\br{A}-\ln\det\br{A}-3)
\end{equation}
denote the local conformation contribution to the free energy. Its contribution to the second-order energetic stress is
\begin{equation}
\br{H}_{1}^{\mathrm{OB}}[\br{A}] \coloneqq \br{A}\partial_{\br{A}}\psi_{\mathrm{OB}} + \partial_{\br{A}}\psi_{\mathrm{OB}}\br{A} = G(\br{A}-\bs{1}),
\end{equation}
which is the classical polymeric extra stress of the Oldroyd-B model. Within the incompressible formulation, \(\bb{P}_0[\br{H}_{1}^{\mathrm{OB}}[\br{A}]]\) appears explicitly in \(\br{S}\), while its spherical part is absorbed into the indeterminate pressure. Thus, for the free energy above, \(\br{H}_{1}[\br{A}]\) contains \(\br{H}_{1}^{\mathrm{OB}}[\br{A}]\) as its local conformation contribution, while its remaining terms provide the gradient-induced generalization.

For the conformation dynamics, we use a state-dependent dissipative response for \(\br{A}\), rather than the scalar drag specialization \eqref{eq:Xi.dis}.
\begin{prop}[Lyapunov mobility and dissipative consistency]\label{prop:L.A}
Let \(\br{A}\in\mathrm{Sym}^+\), and let \(\cl{L}_{\!\br{A}}\colon\mathrm{Sym}\to\mathrm{Sym}\) be the Lyapunov operator defined in \eqref{eq:lyapunov.operator} with \(\br{J}=\br{A}\).
Then the following statements hold.
\begin{enumerate}[label=\textit{(\roman*)}]
\item The family \(\cl{L}_{\!\br{A}}\) is rotationally equivariant, that is,
\begin{equation}\label{eq:L.A.equivariant}
\cl{L}_{\!\br{Q}\br{A}\br{Q}^{\!\trans}}[\br{Q}\br{Z}\br{Q}^{\!\trans}] = \br{Q}\cl{L}_{\!\br{A}}[\br{Z}]\br{Q}^{\!\trans}
\qquad\text{for every }\br{Q}\in\mathrm{SO}(3)\text{ and }\br{Z}\in\mathrm{Sym}.
\end{equation}
For each fixed \(\br{A}\), the operator \(\cl{L}_{\!\br{A}}\) is self-adjoint, invertible, and positive definite on \(\mathrm{Sym}\), with
\begin{equation}\label{eq:L.A.positive}
\br{Z}\twovdots\cl{L}_{\!\br{A}}[\br{Z}] = 2 \mskip2mu \tr(\br{A}\br{Z}^2)>0
\qquad\text{for every }\br{Z}\in\mathrm{Sym}\setminus\{\bs{0}\}.
\end{equation}
\item The dissipative conformation microforce
\begin{equation}\label{eq:Xi.A.OB}
\bs{\Xi}^{\mathrm{dis}}_1 \coloneqq -G\lambda_{\mathrm{p}}\cl{L}_{\!\br{A}}^{-1}[\mathfrak{D}_{\!1}\br{A}], \qquad \lambda_{\mathrm{p}}>0,
\end{equation}
where \(\lambda_{\mathrm{p}}\) is the polymer relaxation time, satisfies
\begin{equation}\label{eq:Xi.A.OB.dissipation}
-\bs{\Xi}^{\mathrm{dis}}_1\twovdots\mathfrak{D}_{\!1}\br{A} = G\lambda_{\mathrm{p}}\mathfrak{D}_{\!1}\br{A}\twovdots\cl{L}_{\!\br{A}}^{-1}[\mathfrak{D}_{\!1}\br{A}]\ge 0.
\end{equation}
\item With \(\bb{X}_1=\partial_{\Grad\br{A}}\psi\) and \(\bs{\Xi}_1=\bs{\Xi}^{\mathrm{dis}}_1-\partial_{\br{A}}\psi\), the conformation microforce balance is equivalent to
\begin{equation}\label{eq:A.OB.gradient.evolution}
\lambda_{\mathrm{p}} \, \mathfrak{D}_{\!1}\br{A} = -\fr{1}{G}\cl{L}_{\!\br{A}}\big[\partial_{\br{A}}\psi-\Div(\partial_{\Grad\br{A}}\psi)-\bs{\Upsilon}_1\big].
\end{equation}
\item For the explicit free energy \eqref{eq:free.energy.OB.LD.explicit}, the conformation driving force is
\begin{equation}\label{eq:A.OB.explicit.force}
\partial_{\br{A}}\psi-\Div(\partial_{\Grad\br{A}}\psi)-\bs{\Upsilon}_1
=\fr{G}{2}(\bs{1}-\br{A}^{-1})-\gamma_1\Delta\br{A}-\bs{\Upsilon}_1.
\end{equation}
Consequently,
\begin{equation}\label{eq:L.A.OB.explicit.force}
\cl{L}_{\!\br{A}}\big[\partial_{\br{A}}\psi-\Div(\partial_{\Grad\br{A}}\psi)-\bs{\Upsilon}_1\big] = G(\br{A}-\bs{1})-\gamma_1\cl{L}_{\!\br{A}}[\Delta\br{A}]-\cl{L}_{\!\br{A}}[\bs{\Upsilon}_1],
\end{equation}
and \eqref{eq:A.OB.gradient.evolution} takes the explicit form
\begin{equation}\label{eq:A.OB.explicit.evolution}
\lambda_{\mathrm{p}} \, \mathfrak{D}_{\!1}\br{A} = -(\br{A}-\bs{1})+\fr{\gamma_1}{G}\cl{L}_{\!\br{A}}[\Delta\br{A}]+\fr{1}{G}\cl{L}_{\!\br{A}}[\bs{\Upsilon}_1].
\end{equation}
\end{enumerate}
\end{prop}
\begin{proof}
\indent\textit{(i)}
Identity \eqref{eq:L.A.equivariant} follows directly from \eqref{eq:lyapunov.operator}. By Lemma \ref{lem:lyapunov.commutator.identities}, with \(\br{J}=\br{A}\), for every \(\br{Z}_1,\br{Z}_2\in\mathrm{Sym}\),
\begin{equation}
\br{Z}_1\twovdots\cl{L}_{\!\br{A}}[\br{Z}_2]
=\cl{L}_{\!\br{A}}[\br{Z}_1]\twovdots\br{Z}_2,
\end{equation}
so \(\cl{L}_{\!\br{A}}\) is self-adjoint. If \(a_i>0\) are the eigenvalues of \(\br{A}\), then, in an orthonormal eigenbasis of \(\br{A}\),
\begin{equation}
(\cl{L}_{\!\br{A}}[\br{Z}])_{ij}=(a_i+a_j)(\br{Z})_{ij}.
\end{equation}
Thus, \(\cl{L}_{\!\br{A}}\) is invertible on \(\mathrm{Sym}\), and \eqref{eq:L.A.positive} proves its positive definiteness.

\medskip
\indent\textit{(ii)}
The inverse of \(\cl{L}_{\!\br{A}}\) is also self-adjoint and positive definite. Substitution of \eqref{eq:Xi.A.OB} therefore proves \eqref{eq:Xi.A.OB.dissipation}.

\medskip
\indent\textit{(iii)}
The conformation microforce balance \(\bs{\Xi}_1+\bs{\Upsilon}_1+\Div\bb{X}_1=\bs{0}\) gives
\begin{equation}
\bs{\Xi}^{\mathrm{dis}}_1=\partial_{\br{A}}\psi-\Div(\partial_{\Grad\br{A}}\psi)-\bs{\Upsilon}_1.
\end{equation}
Combining this identity with \eqref{eq:Xi.A.OB} and applying \(\cl{L}_{\!\br{A}}\) proves \eqref{eq:A.OB.gradient.evolution}.

\medskip
\indent\textit{(iv)}
For the free energy \eqref{eq:free.energy.OB.LD.explicit},
\begin{equation}
\partial_{\br{A}}\psi=\fr{G}{2}(\bs{1}-\br{A}^{-1}) \texand \Div(\partial_{\Grad\br{A}}\psi)=\gamma_1\Delta\br{A},
\end{equation}
which proves \eqref{eq:A.OB.explicit.force}. The identity
\begin{equation}
\cl{L}_{\!\br{A}}\big[\fr{G}{2}(\bs{1}-\br{A}^{-1})\big]=G(\br{A}-\bs{1})
\end{equation}
and the linearity of \(\cl{L}_{\!\br{A}}\) give \eqref{eq:L.A.OB.explicit.force}. Substitution into \eqref{eq:A.OB.gradient.evolution} then yields \eqref{eq:A.OB.explicit.evolution}.
\end{proof}

With \eqref{eq:free.energy.OB.LD} and \eqref{eq:free.energy.OB.LD.explicit} and Propositions \ref{prop:P0} and \ref{prop:L.A}, the field equations read as follows.
\begin{equation}\label{eq:field.equations.OB.LD}
\left\{
\begin{aligned}
& \Div \vel = 0, \\[4pt]
& \varrho \dot{\vel} = \Div (\br{S}-\Div\bb{S})-\Grad\pi+\varrho \bs{b}, \\[4pt]
& \lambda_{\mathrm{p}} \, \mathfrak{D}_{\!1}\br{A} = -\fr{1}{G}\cl{L}_{\!\br{A}}\big[\partial_{\br{A}}\psi-\Div(\partial_{\Grad\br{A}}\psi)-\bs{\Upsilon}_1\big], \\[4pt]
& \beta_0 \, \mathfrak{D}_{\!0}\br{Y} = - \bb{P}_0 [\partial_{\br{Y}}\psi] + \bb{P}_0 [\Div(\partial_{\Grad\br{Y}}\psi)] + \bb{P}_0[\bs{\Upsilon}_0].
\end{aligned}
\right.
\end{equation}
Here,
\begin{equation}
\left\{
\begin{aligned}
\br{S} ={}& 2\mu \br{D} + \bb{P}_0[\br{H}_{1}[\br{A}]+\br{H}_{0}[\br{Y}]], \\[4pt] \bb{S} ={}& \lambda_7 (\Grad\br{L})^{(7)} + \lambda_5 (\Grad\br{L})^{(5)} + \lambda_3 (\Grad\br{L})^{(3)} \\[4pt] &+ \bb{P}_{\!\!\scriptscriptstyle\cl{A}}[\sym_{23}\bb{H}[\br{A}] + \sym_{23}\skw_{\!12}\bb{H}[\br{Y}]], \\[4pt] \bb{X}_1 ={}& \partial_{\Grad\br{A}}\psi, \\[4pt] \bb{X}_0 ={}& \partial_{\Grad\br{Y}}\psi.
\end{aligned}
\right.
\end{equation}
Therefore, in view of the canonical free-energy imbalance \eqref{eq:pointwise.dissipation.final.reduced.equivalent}, the harmonic decomposition \eqref{eq:Q.harmonic.split.theorem}, and \eqref{eq:Xi.A.OB.dissipation}, the dissipation inequality takes the following form
\begin{align}
0 \le{}& 2 \mu |\br{D}|^2 + \lambda_7 |(\Grad\br{L})^{(7)}|^2 + \lambda_5 |(\Grad\br{L})^{(5)}|^2 + \lambda_3 |(\Grad\br{L})^{(3)}|^2 \nonumber\\[4pt]
&+G\lambda_{\mathrm{p}}\mathfrak{D}_{\!1}\br{A}\twovdots\cl{L}_{\!\br{A}}^{-1}[\mathfrak{D}_{\!1}\br{A}] + \beta_0 |\bb{P}_0[\mathfrak{D}_{\!0}\br{Y}]|^2,
\end{align}
with \(\mu\), \(\lambda_7\), \(\lambda_5\), and \(\lambda_3\) non-negative, and with \(G\), \(\lambda_{\mathrm{p}}\), and \(\beta_0\) strictly positive.

For the sake of completeness, the partial derivatives of the free energy density are
\begin{equation}
\left\{
\begin{aligned} & \partial_{\br{A}} \psi = \fr{G}{2} (\bs{1} - \br{A}^{-1}), \\[4pt]
& \partial_{\br{Y}} \psi = a \br{Y} - b \br{Y}^2 + c \tr(\br{Y}^2) \br{Y}, \\[4pt]
& \partial_{\Grad \br{A}} \psi = \gamma_1 \, \Grad \br{A}, \\[4pt]
& \partial_{\Grad \br{Y}} \psi = \gamma_0 \, \Grad \br{Y}.
\end{aligned}
\right.
\end{equation}

\section*{Acknowledgements}

The author acknowledges support from the EPSRC Impact Acceleration Account (IAA) at the University of Nottingham.

\section*{Declarations}

\medskip\noindent\textbf{Funding} This work was supported by the EPSRC Impact Acceleration Account (IAA).

\medskip\noindent\textbf{Data availability} No datasets were generated or analysed during the current study.

\medskip\noindent\textbf{Ethical approval} Not applicable.

\medskip\noindent\textbf{Competing interests} The author declares no competing interests.

%% file: appendix.tex
\section{Geometric Genesis of Transport Laws}\label{ap:transport.rates}

For completeness, we derive the covariant, contravariant, and corotational tensor rates used in Section \S\ref{sc:frame.indifferent.rates} directly from the modes of advection induced by the motion. We follow the terminology of Gurtin, Fried, \& Anand \cite{Gur10} for tangential, covariant, contravariant, and corotational advection. The purpose of this appendix is to expose the geometric genesis of the transport laws by distinguishing tangent and cotangent actions. The main continuum theory, however, is formulated on an ambient Euclidean vector space and follows the standard tensor-space convention of continuum mechanics. The Euclidean metric identifies tensor spaces of corresponding variance through the associated Riesz isomorphisms. Thus, once a transport law has been motivated geometrically, its Euclidean tensor representative is used in the balance laws, power expenditures, and constitutive relations of the main text.

For each material point, \(T_{\bs{x}}\prt_{\scriptscriptstyle 0}\) and \(T_{\bs{y}}\prt_t\) denote the tangent spaces in the reference and current configurations, respectively, while \(T^\ast_{\bs{x}}\prt_{\scriptscriptstyle 0}\) and \(T^\ast_{\bs{y}}\prt_t\) denote the corresponding dual spaces. A spatial vector field \(\bs{a} \in T_{\bs{y}}\prt_t\), a \((1,0)\)-tensor, is said to advect as a tangent if there is a time-independent material vector field \(\bs{a}_{\scriptscriptstyle 0} \in T_{\bs{x}}\prt_{\scriptscriptstyle 0}\) such that
\begin{equation}\label{eq:co-adv.vec}
\bs{a}(\bs{y},t) = \br{F}(\bs{x},t)\bs{a}_{\scriptscriptstyle 0}(\bs{x}).
\end{equation}
The corresponding dual transport of a spatial covector field \(\bs{\alpha} \in T^\ast_{\bs{y}}\prt_t\), a \((0,1)\)-tensor, is defined by a time-independent material covector field \(\bs{\alpha}_{\scriptscriptstyle 0} \in T^\ast_{\bs{x}}\prt_{\scriptscriptstyle 0}\) such that
\begin{equation}\label{eq:contra-adv.covec}
\bs{\alpha}(\bs{y},t) = \br{F}^{-\trans}(\bs{x},t)\bs{\alpha}_{\scriptscriptstyle 0}(\bs{x}).
\end{equation}

In view of \eqref{eq:co-adv.vec} and \eqref{eq:contra-adv.covec},
\begin{equation}
\dot{\overline{\br{F}^{-1}\bs{a}}} = \bs{0} \texand \dot{\overline{\br{F}^{\trans}\bs{\alpha}}} = \bs{0}.
\end{equation}
Thus, a vector advects as a tangent and its dually transported covector evolves if and only if
\begin{equation}
\dot{\bs{a}} = \br{L}\bs{a} \texand \dot{\bs{\alpha}} = -\br{L}^{\!\trans}\bs{\alpha},
\end{equation}
where \(\br{L}\) is a \((1,1)\)-tensor. Consequently, their natural pairing is preserved:
\begin{equation}
\dot{\overline{\langle \bs{\alpha},\bs{a}\rangle}} = \langle \dot{\bs{\alpha}},\bs{a}\rangle + \langle \bs{\alpha},\dot{\bs{a}}\rangle = -\langle \br{L}^{\!\trans}\bs{\alpha},\bs{a}\rangle + \langle \bs{\alpha},\br{L}\bs{a}\rangle = 0.
\end{equation}

Next, consider a vector basis \(\{\bs{m}_\imath\}\) of \(T_{\bs{x}}\prt_{\scriptscriptstyle 0}\) and its dual covector basis \(\{\bs{\mu}^\jmath\}\) of \(T^\ast_{\bs{x}}\prt_{\scriptscriptstyle 0}\). Define
\begin{equation}\label{eq:basis.duality}
T_{\bs{y}}\prt_t \ni \bs{e}_\imath(\bs{y},t) = \br{F}(\bs{x},t)\bs{m}_\imath(\bs{x}) \texand T^\ast_{\bs{y}}\prt_t \ni \bs{\varepsilon}^\jmath(\bs{y},t) = \br{F}^{-\trans}(\bs{x},t)\bs{\mu}^\jmath(\bs{x}).
\end{equation}
The family \(\{\bs{e}_\imath\}\) forms a tangentially advecting basis, while \(\{\bs{\varepsilon}^\jmath\}\) is its dual basis. These satisfy
\begin{equation}
\langle\bs{\varepsilon}^\jmath,\bs{e}_\imath\rangle = \delta^\jmath_\imath.
\end{equation}
Since \(\dot{\bs{m}}_\imath=\bs{0}\) and \(\dot{\bs{\mu}}^\jmath=\bs{0}\),
\begin{equation}\label{eq:co-adv.basis}
\dot{\bs{e}}_\imath = \br{L}\bs{e}_\imath \texand \dot{\bs{\varepsilon}}^\jmath = -\br{L}^{\!\trans}\bs{\varepsilon}^\jmath.
\end{equation}

Let
\begin{equation}
\br{H} \in T^\ast_{\bs{y}}\prt_t \otimes T^\ast_{\bs{y}}\prt_t
\end{equation}
be a spatial covariant second-order tensor, that is, a \((0,2)\)-tensor satisfying
\begin{equation}
\br{H}(\bs{a},\bs{b})\in\bb{R} \qquad \text{for all } \bs{a},\bs{b}\in T_{\bs{y}}\prt_t.
\end{equation}
Its covariant components relative to the tangentially advecting basis are
\begin{equation}
H_{\imath\jmath} \coloneqq \br{H}(\bs{e}_\imath,\bs{e}_\jmath).
\end{equation}
Using \eqref{eq:co-adv.basis}\(_1\), total time differentiation gives
\begin{align}\label{eq:total.derivative.covariant.1}
\dot{H}_{\imath\jmath}
&= \dot{\overline{\br{H}(\bs{e}_\imath,\bs{e}_\jmath)}} \nonumber\\[4pt]
&= \dot{\br{H}}(\bs{e}_\imath,\bs{e}_\jmath) + \br{H}(\dot{\bs{e}}_\imath,\bs{e}_\jmath) + \br{H}(\bs{e}_\imath,\dot{\bs{e}}_\jmath) \nonumber\\[4pt]
&= \dot{\br{H}}(\bs{e}_\imath,\bs{e}_\jmath) + \br{H}(\br{L}\bs{e}_\imath,\bs{e}_\jmath) + \br{H}(\bs{e}_\imath,\br{L}\bs{e}_\jmath).
\end{align}
Define
\begin{equation}
(\br{L}^{\!\trans}\br{H})(\bs{u},\bs{v}) \coloneqq \br{H}(\br{L}\bs{u},\bs{v}) \texand (\br{H}\br{L})(\bs{u},\bs{v}) \coloneqq \br{H}(\bs{u},\br{L}\bs{v}).
\end{equation}
Then,
\begin{equation}\label{eq:total.derivative.covariant.1.final}
\dot{H}_{\imath\jmath} = \big(\dot{\br{H}} + \br{L}^{\!\trans}\br{H} + \br{H}\br{L}\big)(\bs{e}_\imath,\bs{e}_\jmath).
\end{equation}
Equivalently,
\begin{equation}\label{eq:total.derivative.covariant.2}
\dot{H}_{\imath\jmath} = \dot{\overline{(\br{F}^{\trans}\br{H}\br{F})(\bs{m}_\imath,\bs{m}_\jmath)}},
\end{equation}
where
\begin{equation}
(\br{F}^{\trans}\br{H}\br{F})(\bs{a},\bs{b}) \coloneqq \br{H}(\br{F}\bs{a},\br{F}\bs{b}).
\end{equation}
Thus, the spatial tensor field \(\br{H}\) is said to advect covariantly if
\begin{equation}
\underset{\triangledown}{\br{H}} \coloneqq \dot{\br{H}} + \br{L}^{\!\trans}\br{H} + \br{H}\br{L} = \bs{0},
\end{equation}
or, equivalently, if
\begin{equation}
\dot{\overline{\br{F}^{\trans}\br{H}\br{F}}} = \bs{0}.
\end{equation}
The covariant components of \(\br{H}\) relative to the tangentially advecting basis are then materially time-independent. The rate \(\underset{\triangledown}{\br{H}}\) is the covariant rate of \(\br{H}\), also referred to as its lower-convected rate.

Similarly, let
\begin{equation}
\br{H} \in T_{\bs{y}}\prt_t \otimes T_{\bs{y}}\prt_t
\end{equation}
be a spatial contravariant second-order tensor, that is, a \((2,0)\)-tensor satisfying
\begin{equation}
\br{H}(\bs{\alpha},\bs{\beta}) \in \bb{R} \qquad \text{for all } \bs{\alpha},\bs{\beta} \in T^\ast_{\bs{y}}\prt_t.
\end{equation}
Its contravariant components relative to the dual basis are
\begin{equation}
H^{\imath\jmath} \coloneqq \br{H}(\bs{\varepsilon}^\imath,\bs{\varepsilon}^\jmath).
\end{equation}
Using \eqref{eq:co-adv.basis}\(_2\), total time differentiation gives
\begin{align}\label{eq:total.derivative.contravariant.1}
\dot{H}^{\imath\jmath}
&= \dot{\overline{\br{H}(\bs{\varepsilon}^\imath,\bs{\varepsilon}^\jmath)}} \nonumber\\[4pt]
&= \dot{\br{H}}(\bs{\varepsilon}^\imath,\bs{\varepsilon}^\jmath) + \br{H}(\dot{\bs{\varepsilon}}^\imath,\bs{\varepsilon}^\jmath) + \br{H}(\bs{\varepsilon}^\imath,\dot{\bs{\varepsilon}}^\jmath) \nonumber\\[4pt]
&= \dot{\br{H}}(\bs{\varepsilon}^\imath,\bs{\varepsilon}^\jmath) - \br{H}(\br{L}^{\!\trans}\bs{\varepsilon}^\imath,\bs{\varepsilon}^\jmath) - \br{H}(\bs{\varepsilon}^\imath,\br{L}^{\!\trans}\bs{\varepsilon}^\jmath).
\end{align}
Define
\begin{equation}
(\br{L}\br{H})(\bs{\alpha},\bs{\beta}) \coloneqq \br{H}(\br{L}^{\!\trans}\bs{\alpha},\bs{\beta}) \texand (\br{H}\br{L}^{\!\trans})(\bs{\alpha},\bs{\beta}) \coloneqq \br{H}(\bs{\alpha},\br{L}^{\!\trans}\bs{\beta}).
\end{equation}
Then,
\begin{equation}
\dot{H}^{\imath\jmath} = \big(\dot{\br{H}} - \br{L}\br{H} - \br{H}\br{L}^{\!\trans}\big)(\bs{\varepsilon}^\imath,\bs{\varepsilon}^\jmath).
\end{equation}
Equivalently,
\begin{equation}\label{eq:total.derivative.contravariant.2}
\dot{H}^{\imath\jmath} = \dot{\overline{(\br{F}^{-1}\br{H}\br{F}^{-\trans})(\bs{\mu}^\imath,\bs{\mu}^\jmath)}},
\end{equation}
where
\begin{equation}
(\br{F}^{-1}\br{H}\br{F}^{-\trans})(\bs{\alpha},\bs{\beta}) \coloneqq \br{H}(\br{F}^{-\trans}\bs{\alpha},\br{F}^{-\trans}\bs{\beta}).
\end{equation}
Thus, the spatial tensor field \(\br{H}\) is said to advect contravariantly if
\begin{equation}
\overset{\triangledown}{\br{H}} \coloneqq \dot{\br{H}} - \br{L}\br{H} - \br{H}\br{L}^{\!\trans} = \bs{0},
\end{equation}
or, equivalently, if
\begin{equation}
\dot{\overline{\br{F}^{-1}\br{H}\br{F}^{-\trans}}} = \bs{0}.
\end{equation}
The contravariant components of \(\br{H}\) relative to the dual basis are then materially time-independent. The rate \(\overset{\triangledown}{\br{H}}\) is the contravariant rate of \(\br{H}\), also referred to as its upper-convected rate.

Lastly, a spatial vector field \(\bs{k} \in T_{\bs{y}}\prt_t\) is said to advect corotationally if
\begin{equation}\label{eq:co-rot.basis}
\dot{\bs{k}} = \br{W}\bs{k},
\end{equation}
where
\begin{equation}
\br{W} \coloneqq \skw\br{L} \in T_{\bs{y}}\prt_t \otimes T^\ast_{\bs{y}}\prt_t.
\end{equation}
The corresponding dual covector field \(\bs{\kappa} \in T^\ast_{\bs{y}}\prt_t\) evolves according to
\begin{equation}
\dot{\bs{\kappa}} = -\br{W}^{\!\trans}\bs{\kappa}.
\end{equation}
Consequently,
\begin{equation}
\dot{\overline{\langle \bs{\kappa},\bs{k}\rangle}} = \langle \dot{\bs{\kappa}},\bs{k}\rangle + \langle \bs{\kappa},\dot{\bs{k}}\rangle = 0.
\end{equation}

We now restrict the contravariant tensorial representation used above to corotational advection. Let
\begin{equation}
\br{H} \in T_{\bs{y}}\prt_t \otimes T_{\bs{y}}\prt_t.
\end{equation}
The total time derivative of its contravariant components relative to the dual basis of a corotationally advecting vector basis gives
\begin{equation}
\dot{\overline{\br{H}(\bs{\kappa}^\imath,\bs{\kappa}^\jmath)}} = \big(\dot{\br{H}} - \br{W}\br{H} - \br{H}\br{W}^{\!\trans}\big)(\bs{\kappa}^\imath,\bs{\kappa}^\jmath).
\end{equation}
Since \(\br{W}^{\!\trans}=-\br{W}\), this reduces to
\begin{equation}
\overset{\circ}{\br{H}} \coloneqq \dot{\br{H}} - \br{W}\br{H} + \br{H}\br{W}.
\end{equation}
The rate \(\overset{\circ}{\br{H}}\) is the corotational, or Jaumann, rate of \(\br{H}\).

The differential equation \eqref{eq:co-rot.basis} induces a one-parameter family of orthogonal transformations \(\br{Q}(t)\in\mathrm{SO}(3)\) along each material trajectory. Let
\begin{equation}
\dot{\br{Q}} = \br{W}\br{Q}, \with \br{Q}(t_{\scriptscriptstyle 0}) = \bs{1}.
\end{equation}
Since \(\br{W}^{\!\trans}=-\br{W}\),
\begin{equation}
\dot{\overline{\br{Q}^{\!\trans}\br{Q}}} = \dot{\br{Q}}^{\!\trans}\br{Q} + \br{Q}^{\!\trans}\dot{\br{Q}} = \br{Q}^{\!\trans}(\br{W}^{\!\trans}+\br{W})\br{Q} = \bs{0},
\end{equation}
and hence
\begin{equation}
\br{Q}^{\!\trans}\br{Q} = \bs{1}.
\end{equation}

Accordingly, a vector field advects corotationally if and only if
\begin{equation}
\bs{k}(\bs{y},t) = \br{Q}(t)\bs{k}_{\scriptscriptstyle 0}(\bs{x}),
\end{equation}
for some time-independent material vector field \(\bs{k}_{\scriptscriptstyle 0}\in T_{\bs{x}}\prt_{\scriptscriptstyle 0}\). The corresponding dual covector satisfies
\begin{equation}
\bs{\kappa}(\bs{y},t) = \br{Q}^{-\trans}(t)\bs{\kappa}_{\scriptscriptstyle 0}(\bs{x}).
\end{equation}
A contravariant second-order tensor advects corotationally according to
\begin{equation}
\br{H}(\bs{y},t) = \br{Q}(t)\br{H}_{\scriptscriptstyle 0}(\bs{x})\br{Q}(t)^{\!\trans},
\end{equation}
and consequently satisfies
\begin{equation}
\overset{\circ}{\br{H}} = \dot{\br{H}} - \br{W}\br{H} + \br{H}\br{W} = \bs{0}.
\end{equation}

\section{Tangent spaces of \texorpdfstring{\(\mathrm{GL}^+(3)\) and \(\mathrm{SO}(3)\)}{GL+(3) and SO(3)}}\label{ap:alg.grp}

The Lie algebras \(\mathfrak{gl}(3)\) and \(\mathfrak{so}(3)\) coincide with the tangent spaces of \(\mathrm{GL}^+(3)\) and \(\mathrm{SO}(3)\) at the identity, that is,
\begin{equation}
\mathfrak{gl}(3) \coloneqq T_{\bs{1}}\mathrm{GL}^+(3) = \bb{R}^{3\times 3} \texand \mathfrak{so}(3) \coloneqq T_{\bs{1}}\mathrm{SO}(3) = \{\bs{\Omega} \in \bb{R}^{3\times 3} \colon \bs{\Omega}^{\!\trans} = -\bs{\Omega} \}.
\end{equation}

We first consider \(\mathrm{GL}^+(3)\). Since \(\mathrm{GL}^+(3)\) is an open subset of \(\bb{R}^{3\times 3}\), its tangent space at the identity coincides with the ambient space, that is,
\begin{equation}
T_{\bs{1}}\mathrm{GL}^+(3) = \bb{R}^{3\times 3} = \mathfrak{gl}(3).
\end{equation}

Next, we consider \(\mathrm{SO}(3)\). Let \(\br{Q}(\epsilon)\in \mathrm{SO}(3)\) be a smooth curve such that \(\br{Q}(0)=\bs{1}\). Since
\begin{equation}
\br{Q}(\epsilon)^{\!\trans}\br{Q}(\epsilon)=\bs{1},
\end{equation}
differentiation yields
\begin{equation}
\bigg(\dd{\br{Q}}{\epsilon}\bigg)^{\!\!\trans}\br{Q} + \br{Q}^{\!\trans}\dd{\br{Q}}{\epsilon} = \bs{0}.
\end{equation}
Evaluating at \(\epsilon=0\) gives
\begin{equation}
\bigg(\dd{\br{Q}}{\epsilon}(0)\bigg)^{\!\!\trans} + \dd{\br{Q}}{\epsilon}(0) = \bs{0},
\end{equation}
so that every tangent vector \(\bs{\Omega} \in T_{\bs{1}}\mathrm{SO}(3)\) satisfies
\begin{equation}
\bs{\Omega}^{\!\trans} = -\bs{\Omega}.
\end{equation}
Hence
\begin{equation}
T_{\bs{1}}\mathrm{SO}(3) \subseteq \mathfrak{so}(3).
\end{equation}

Conversely, let \(\bs{\Omega}\in \mathfrak{so}(3)\). Then, the curve
\begin{equation}
\br{Q}(\epsilon)=\exp(\epsilon \bs{\Omega}),
\end{equation}
satisfies \(\br{Q}(0)=\bs{1}\),
\begin{equation}
\br{Q}(\epsilon)^{\!\trans}\br{Q}(\epsilon) = \exp(\epsilon \bs{\Omega})^{\!\trans}\exp(\epsilon \bs{\Omega}) = \exp(\epsilon \bs{\Omega}^{\!\trans})\exp(\epsilon \bs{\Omega}) = \exp(-\epsilon \bs{\Omega})\exp(\epsilon \bs{\Omega}) = \bs{1},
\end{equation}
and
\begin{equation}
\det \br{Q}(\epsilon) = \det \exp(\epsilon \bs{\Omega}) = \exp(\epsilon \tr \bs{\Omega}) = 1,
\end{equation}
because \(\tr \bs{\Omega}=0\) for every skew-symmetric tensor. Therefore, \(\br{Q}(\epsilon)\in \mathrm{SO}(3)\) and
\begin{equation}
\dd{\br{Q}}{\epsilon}(0)=\bs{\Omega}.
\end{equation}
Thus, \(\bs{\Omega}\in T_{\bs{1}}\mathrm{SO}(3)\), and therefore
\begin{equation}
T_{\bs{1}}\mathrm{SO}(3)=\mathfrak{so}(3).
\end{equation}